\documentclass[12pt]{article}

\RequirePackage{amsthm,amsmath,amsfonts,amssymb}
\RequirePackage[numbers]{natbib}
\RequirePackage[colorlinks,citecolor=blue,urlcolor=blue]{hyperref}
\RequirePackage{graphicx}

\usepackage{array}

\theoremstyle{plain}

\newtheorem{theorem}{Theorem}[section]
\newtheorem{lemma}[theorem]{Lemma}
\theoremstyle{remark}
\newtheorem{definition}[theorem]{Definition}

\newtheorem{cor}[theorem]{Corollary}

\def\e{\varepsilon}

\newcommand{\bm}[1]{ \mathbf{#1} }
\def\matA{{\bm A}}
\def\matB{{\bm B}}
\def\matC{{\bm C}}
\def\matI{{\bm I}}

\def\matM{{\bm M}}

\def\matQ{{\bm Q}}
\def\matR{{\bm R}}
\def\matV{{\bm V}}

\def\matX{{\bm X}}
\def\matY{{\bm Y}}
\def\matZ{{\bm Z}}

\def\mate{{\bm e}}
\def\matb{{\bm b}}
\def\matc{{\bm c}}
\def\matv{{\bm v}}
\def\matx{{\bm x}}
\def\maty{{\bm y}}

\def\matDelta{{\bm \Delta}}
\def\mat0{{\bm 0}}

\newcommand{\be}{\begin{equation}}
\newcommand{\ee}{\end{equation}}
\newcommand{\ba}{\begin{array}}
\newcommand{\ea}{\end{array}}

\begin{document}

\begin{titlepage}

\renewcommand{\thepage}{}
\title{Reducing Noise Level in Differential Privacy through Matrix Masking}

\author{A. Adam Ding, Samuel S. Wu, Guanhong Miao and Shigang Chen}

\date{}
\maketitle
\begin{abstract}
Differential privacy schemes have been widely adopted in recent years to address issues of data privacy protection.
We propose a new Gaussian scheme combining with another data protection technique, called random orthogonal matrix masking, to achieve $(\e, \delta)$-differential privacy (DP) more efficiently.
We prove that the additional matrix masking significantly reduces the rate of noise variance required in the Gaussian scheme to achieve $(\e, \delta)-$DP in big data setting.
Specifically, when $\e \to 0$, $\delta \to 0$, and the sample size $n$ exceeds the number $p$ of attributes by $(n-p)=O(ln(1/\delta))$, the required additive noise variance to achieve $(\e, \delta)$-DP is reduced from $O(ln(1/\delta)/\e^2)$ to $O(1/\e)$.
With much less noise added, the resulting differential privacy protected pseudo data sets allow much more accurate inferences, thus can significantly improve the scope of application for differential privacy.
\end{abstract}

\end{titlepage}


\section{Introduction}

In the past decades, much research interest has been devoted to the issue of how to share and publish data sets while  protecting the privacy of individuals in the data set.
The differential privacy (DP)~\citep{dwork2006differential,dwork2006CalibratingNoise,dwork2008differential} provides a quantitative measure for privacy loss in data release. In brief, DP guarantees that the distributions over potential outputs are statistically close for any two neighboring databases, and thus the chance of an individual being identified is low.  Over the years, DP schemes have been widely adopted, through adding noises, to release data statistics with given privacy cost constraints. For instance, DP has been implemented in real world applications developed by Google~\citep{google}, Microsoft~\citep{microsoft} and the U.S. Census Bureau~\citep{Bureau}.

In recent literature, there has been an emphasis on linking DP to statistical concepts. Objective perturbation was applied to build differentially private linear regression and logistic regression model~\citep{DP_3,DP_4}. Another well-known method, Johnson-Lindenstrauss transform, has been investigated to provide DP by mapping original data to a lower-dimensional space~\citep{DP_1,DP_2}. The DP framework has been expanded to include hypothesis tests~\citep{DP_6,DP_7}, deep learning~\citep{DP_5}, network analysis~\citep{DP_network} and Bayesian inference~\citep{DP_8}. Most DP schemes focus on releasing some specific summary statistics rather than a whole pseudo data set (i.e., a data set with noise perturbations of its entries), which is the focus of this work.
Most recently the concept of local differential privacy (LDP)~\citep{LDP} has come to the fore where users perturb their data and then send the perturbed data to an aggregator for summary and release. The aggregator does not know the raw user data. LDP provides a stronger privacy guarantee than the central DP approach where users directly share their raw data with a trusted aggregator that performs perturbation. A major limitation of LDP is that decentralized perturbation causes larger error in order to achieve the same privacy level.

In another thread of research, random orthogonal perturbation has been deployed for privacy protection~\citep{perturbation,AX,wu2017new} where a masked data set $\matA\matX$ is published, with $\matX$ denoting the original raw data set, and $\matA$ being a random orthogonal matrix.
This pseudo data set $\matY=\matA\matX$ allows users to use many standard data analysis methods since, for linear models, it has the same sufficient statistics as the raw data set.
However, the random orthogonal matrix masking schemes do not satisfy the theoretical DP constraints.

In this paper, we investigate combining the random orthogonal matrix masking with noise addition methods to achieve DP on the released pseudo data set. Compared to the traditional method of adding noise directly to the raw data set, we show that the matrix masking greatly reduces the magnitude of noise needed to achieve DP, thus allowing more accurate statistical inferences on the released pseudo data set. Moreover, data utility does not decrease significantly due to the property of the orthogonal masking matrix and the small magnitude of noise required for DP. This opens up many more applications for using differential privacy schemes.

\section{Mathematical Setup for Collecting and Releasing Pseudo-Data Set Satisfying Differential Privacy}

Mathematically, a raw data set $\matX$ is represented as an  $n \times p$ matrix where the columns are the features (variables) and the rows correspond to individuals in the data.
Here we consider the problem of collecting a pseudo data set, represented as an $n \times p$ matrix $\matY$, which is a perturbed version of raw data $\matX$. The goal of privacy preserving perturbation mechanism is to generate $\matY$ such that, while keeping some population statistical information similar to $\matX$ for data analysis utility, makes it hard to inference the raw data for each individual (each row of  $\matX$) from $\matY$. Differential privacy (DP) formalizes the privacy requirement as that, if two raw data sets $\matX$ and $\matX'$ only differs in one individual (one row of the matrix), then it is hard to distinguish whether $\matY$ was generated from $\matX$ or from $\matX'$. Mathematically, the DP requirements can be specified as follows.
\begin{definition}\label{def:Neighbors}
Two data sets $\matX$ and $\matX'$ are neighbors if $\|\matX - \matX' \| \le 1$ and they differ only in one entry (one row).\footnote{Some literature defines the difference between two neighbors as one of them missing an entry that the other has~\citep{dwork2008differential}. Here we define neighbors as that they differ in one entry and the norm of the difference is bounded by one~\citep{dwork2014algorithmic}. For technical simplicity, we use the $L_2-$norm in this paper. Other norms were also used in literature but various $L_p-$norms can be bounded with each other mathematically.} Here and following, $\| \cdot \|$ denotes the $L_2-$norm.
\end{definition}

\begin{definition}\label{def:DiffPriv}
A data perturbation mechanism $\matY$ satisfies the $(\e, \delta)$-differential privacy \citep{DP_0} if, for any set $\mathcal{S}$ and any pair of neighbors $\matX$ and $\matX'$,
\be\label{eq:DP}
P_\matX[\matY \in \mathcal{S}] \le e^\e P_{\matX'}[\matY \in \mathcal{S}] + \delta.
\ee
\end{definition}
Here and in the following, $P_\matX$ denotes the probability when the raw data set is $\matX$. Also, we will use the shorthand notation $(\e, \delta)$-DP to denote $(\e, \delta)$-differential privacy in the rest of the paper.

The differential privacy literature generally focus on getting a privacy-protected statistic $F(\matY)$ which is used for a pre-specified statistical model analysis. That is, the requirement in the DP Definition~\ref{def:DiffPriv} is instead $P_\matX[F(\matY) \in \mathcal{S}] \le e^\e P_{\matX'}[F(\matY) \in \mathcal{S}] + \delta$.
Here we want to collect a pseudo data set not specific for only one data analysis, but to allow more general data analysis by users in the future. Therefore, no statistic $F$ is specified, and the DP Definition~\ref{def:DiffPriv} is needed for privacy protection starting at the collection stage. However, achieving DP Definition~\ref{def:DiffPriv} often requires much higher magnitude of noise than achieving DP only for a specific statistic $F$, thus reducing the required noise magnitude is very important for the wide application of  such a privacy preserving data collection scheme in practice.

Initially the differential privacy protection was used for queries to a central data manager who have access to the raw data $\matX$. For privacy preserving data collection, we want to remove the role of a central data manager so that every party has access to only the pseudo data set $\matY$ and its own data. This can be achieved through a local differential privacy scheme where each individual provides perturbed version of its own data to the data collector. Let $\matX_i$ and $\matY_i$ denote the $i$-th row of respectively $\matX$ and $\matY$ so that
\be\label{eq:rows}
\matX = \left(\ba{c} \matX_1 \\ \matX_2 \\ \vdots \\ \matX_n \ea \right),  \qquad \matY = \left(\ba{c} \matY_1 \\ \matY_2 \\ \vdots \\ \matY_n \ea \right).
\ee
Then in a local $(\e, \delta)$-DP scheme, $\matY_i$ only depends on the raw data $\matX_i$ of the $i$-th individual but not other individual's raw data. Such a local $(\e, \delta)$-DP scheme is generally achieved through noise addition $\matY_i = \matX_i + \matC_i$. The noise $\matC_i$ is most often generated as either Laplace or Gaussian noise. The advantage of the Laplace noise is that it achieves $(\e, 0)$-DP with $\delta=0$. For multi-variate ($p>1$) case, Gaussian noise is often used since it can achieve the $(\e, \delta)$-DP with $L_2-$norm while Laplace noise requires usage of $L_1-$norm. We consider such a Gaussian local $(\e, \delta)$-DP scheme as the baseline from which we improve with matrix masking.

Setting (A) -- Add noise to achieve $(\e, \delta)$-DP: release
\be\label{eq:modelA}
\matY =  \matX + \matC \qquad \mbox{ where } \matC \sim NI_{n \times p}(0, \sigma^2).
\ee
$NI_{n \times p}(0, \sigma^2)$ denotes an $n \times p$ matrix whose elements are independently identically distributed $N(0, \sigma^2)$ random variables.

In contrast, we propose to combine the noise addition with masking by a random orthogonal $n \times n$ matrix $\matA$ to achieve $(\e, \delta)$-DP as in the following setting.

(B) Add noise and apply matrix masking: release
\be\label{eq:modelB}
\matY =  \matA (\matX + \matC) \qquad \mbox{ where } \matC \sim NI_{n \times p}(0, \sigma^2).
\ee
Each individual still adds Gaussian noise $\matC_i$ to its raw data $\matX_i$. Then using the triple-matrix masking scheme~\cite{ding2020privacy} to collect a masked pseudo data set $\matA (\matX + \matC)$ from all data providers without allowing any party in the process to gain more useful information other than the published final $\matA (\matX + \matC)$.

We will show that under the setting (B) with matrix masking, the noise magnitude $\sigma^2$ needed to achieve $(\e, \delta)$-DP can be much smaller than that needed under the noise only setting (A). First, it is obvious that an $(\e, \delta)$-DP scheme under setting (A) will remain $(\e, \delta)$-DP with the extra matrix masking in setting (B). 
We state this in the following lemma whose proof is provided in subsection~\ref{sec:proof.BC<A}.
\begin{lemma}\label{lem:BC<A}
Assume that, for a given noise variance value $\sigma=\sigma_0$, the mechanism in setting (A) is $(\e, \delta)$-DP. Then for this value $\sigma=\sigma_0$, the mechanism in setting (B) is also $(\e, \delta)$-DP.
\end{lemma}

While the above Lemma states that the introduction of the extra matrix masking will never require a stronger noise variance condition to achieve $(\e, \delta)$-DP, we want to explore when the mechanisms in setting (B) with matrix masking will require a weaker condition to achieve $(\e, \delta)$-DP. To study this, we first state another mathematical condition for a release mechanism $Y(\matX)$ to satisfy $(\e, \delta)$-DP. (The proof is provided in subsection~\ref{sec:proof.DPdensity})
\begin{lemma}\label{lem:DPdensity}
Assume that, given any raw data set $\matX$, $Y$ follows a distribution that is absolutely continuous with respect to the Lebesgue measure. And let $p_{\matX}$ denote the probability density of $Y$ when the raw data set is $\matX$. Denote
$$
\mathcal{S}_{\matX,\matX'} = \{\maty: p_{\matX}(\maty) > e^\e p_{\matX'}(\maty)\}.
$$
Then the mechanism satisfies $(\e, \delta)$-DP if
\be\label{eq:DP-bound0}
P_\matX[Y \in \mathcal{S}_{\matX,\matX'}] \le \delta
\ee
for all neighboring $\matX$ and $\matX'$.

Also, if $P_\matX[Y \in \mathcal{S}_{\matX,\matX'}] > \delta$ for any neighboring $\matX$ and $\matX'$, then the mechanism is not $(\e', \delta')$-DP for any $0<\e'<\e$ and $\delta' =(1-\frac{e^{\e'}}{e^\e})\delta$.
\end{lemma}

We first derive sufficient and necessary conditions on the noise magnitude $\sigma$ for achieving ~\eqref{eq:DP-bound0} in setting (A), and the sufficient and necessary conditions are very close to each other. Then we show that achieving $(\e, \delta)$-DP in setting (B) can require much smaller $\sigma$ than those required by the necessary condition for ~\eqref{eq:DP-bound0} in setting (A). Thus Lemma~\ref{lem:DPdensity} implies that it is indeed easier to achieve $(\e, \delta)$-DP in setting (B).


\section{Main Analysis Results}\label{sec:results}

For the technical analysis, we assume that the elements of the raw data $\matX$ are bounded and we can scale them to within magnitude of one. Let $x_{ij}$ denote the $j$-th element of $i$-th row $\matX_i$ in $\matX$. That is, we assume that
\be\label{eq:XboundedOne}
|x_{ij}| \le 1, \qquad \mbox{ for $1 \le i \le n$ and $1\le j \le p$.}
\ee

The following Lemma gives the formulas for the density ratio $\frac {p_{\matX}(\maty)}{p_{\matX'}(\maty)}$ under settings (A) and (B), which allows the characterization of $\mathcal{S}_{\matX,\matX'}$.

\begin{lemma}\label{lem:DenRatio}
\be\label{eq:denRatioA0}
\frac {p_{\matX}(\maty)}{p_{\matX'}(\maty)}  = e^{\frac{\|\matX'\|^2 - \|\matX\|^2 }{2 \sigma^2}} \frac {e^{\frac{tr(\maty\matX^T)}{ \sigma^2}}}{e^{\frac{tr[\maty(\matX')^T] }{\sigma^2}}},
\mbox{ under setting (A);}
\ee
Let $\mathcal{O}_{n \times n}$ denote the group of $n \times n$ orthogonal matrices, and let $\mu(\cdot)$ be the measure for the uniform distribution on $\mathcal{O}_{n \times n}$, then
\be\label{eq:denRatioB0}
\frac {p_{\matX}(\maty)}{p_{\matX'}(\maty)} = e^{\frac{\|\matX'\|^2 - \|\matX\|^2 }{2 \sigma^2}} \frac{\int_{\matA \in \mathcal{O}_{n \times n}} e^{\frac{ tr(\matA^T\maty\matX^T)}{ \sigma^2}} d \mu(\matA)}{\int_{\matA \in \mathcal{O}_{n \times n}} e^{\frac{tr[\matA^T\maty(\matX')^T]}{\sigma^2}} d \mu(\matA)},
\mbox{ under setting (B). } 
\ee
\end{lemma}

Comparing \eqref{eq:denRatioA0} and \eqref{eq:denRatioB0}, the matrix masking averages the densities, before taking the ratio, over the set of all data points that differ only by multiplication of an orthogonal matrix.
This makes the density ratio in equation \eqref{eq:denRatioB0} closer to $1$ than the ratio in equation \eqref{eq:denRatioA0}, thus making it easier to achieve $(\e, \delta)$-DP in setting (B).

The proofs of Lemma~\ref{lem:DenRatio} as well as proofs of main Theorem~\ref{thm:bound.settingA} and Theorem~\ref{thm:bound.settingB_Ver3} stated in Section~\ref{sec:results} will be provided later in Section~\ref{sec:proof}.

\subsection{Noise magnitude to achieve \eqref{eq:DP-bound0} under setting (A)}
We first study the setting (A) when there is only additive noise with no matrix masking.

Denote $\matDelta = \matX ' - \matX$. And let $\bar \gamma_{\delta}$ denote the upper $\delta$-quantile of the standard Gaussian distribution $N(0,1)$. Using \eqref{eq:denRatioA0}, we get a lower bound of $\sigma$ to achieve \eqref{eq:DP-bound0} under setting (A) as the following. 
\begin{theorem}\label{thm:bound.settingA}
Under setting (A), when $\delta<1/2$ and $\e <1$, a necessary condition for \eqref{eq:DP-bound0} to hold is
\be\label{eq:DP-bound0A.nec}
\sigma \ge \frac{\| \matDelta\| \bar \gamma_{\delta}}{\e},
\ee
and a sufficient condition for \eqref{eq:DP-bound0} to hold is
\be\label{eq:DP-bound0A.suf}
\sigma \ge \frac{\| \matDelta\| \bar \gamma_{\delta}}{\e}(1+\frac{1}{2 \bar \gamma_{\delta}^2}).
\ee
Hence for \eqref{eq:DP-bound0} to hold for every pair of neighboring $\matX$ and $\matX'$ (where $\|\matDelta\| \le 1$), it is necessary that
\be\label{eq:DP-boundA.nec}
\sigma \ge \frac{  \bar \gamma_{\delta}}{\e};
\ee
and it is sufficient that
\be\label{eq:DP-boundA.suf}
\sigma \ge \frac{ \bar \gamma_{\delta}}{\e}(1+\frac{1}{2 \bar \gamma_{\delta}^2}).
\ee
\end{theorem}

We note that the upper $\delta$-quantile of the standard Gaussian distribution $\bar \gamma_{\delta}$ is of order $O(\sqrt{ln(\frac{1}{\delta})})$. Thus Theorem~\ref{thm:bound.settingA} states that, for \eqref{eq:DP-bound0} to hold for every pair of neighboring $\matX$ and $\matX'$, $\sigma \ge \sigma_0$ for some $\sigma_0= O(\frac{\sqrt{\ln(\frac{1}{\delta})}}{\e}) $ which agrees with the order in the sufficient condition derived by~\cite{dwork2014algorithmic}. Our Theorem~\ref{thm:bound.settingA} additionally shows that a necessary condition is really close to the sufficient condition, differing by a factor of $(1+\frac{1}{2 \bar \gamma_{\delta}^2}) = O(1+1/\sqrt{ln(\frac{1}{\delta})}) = O(1)$.
Particularly, we can have the following explicit bounds for small $\delta$ and $\e$ values.

\begin{cor}\label{cor:bound.settingA}
Under setting (A), when $\delta<0.05$ and $\e <1$, a sufficient condition for \eqref{eq:DP-bound0} to hold for every pair of neighboring $\matX$ and $\matX'$ is
$\sigma > \frac{1.7  \sqrt{\ln(\frac{1}{\delta})}}{\e}$;
a necessary condition for \eqref{eq:DP-bound0} to hold for every pair of neighboring $\matX$ and $\matX'$ is
$\sigma \ge \frac{\sqrt{\ln(\frac{1}{\delta})}}{\e}$.
\end{cor}

\subsection{Noise magnitude to achieve $(\e, \delta)$-DP under settings (B)} 

In contrast, we have the following results for achieving $(\e, \delta)$-DP with a random matrix masking.

\begin{theorem}\label{thm:bound.settingB_Ver3}
Under settings (B), a sufficient condition for the mechanism in \eqref{eq:modelB} to achieve $(\e, \delta)$-DP is
\be\label{eq:sigma.bound.new}
\ba{cl}
\sigma \ge \sigma_0, & \mbox{ where $\sigma_0$ is the largest root of  $g(\sigma) = 0$ for the function } \\
& g(\sigma) = \frac{(2 \sqrt{p} +1)  }{2 (n-p) } \gamma_{\delta; 2(n-p), \frac{p}{\sigma^2}}  + \sqrt{p} - \sigma^2 \e,
\ea
\ee
where $\gamma_{\delta; df, a}$ denotes the upper $\delta$-quantile of the non-central $\chi^2$ distribution with $df$ degrees of freedom and non-central parameter $a$.
\end{theorem}

We note here that $g(\sigma)$ is a continuous function with limits $\lim_{\sigma \to \infty} g(\sigma) = - \infty$ and $\lim_{\sigma \to 0} g(\sigma) = + \infty$. Hence there always exists a largest positive root $\sigma_0$, and $g(\sigma) \le 0$ for all $\sigma \ge \sigma_0$.

To understand the condition \eqref{eq:sigma.bound.new} better, we notice that $\gamma_{\delta; 2(n-p), \frac{p}{\sigma^2}} \approx 2(n-p)$ for large values of $n-p$. Using the bound for $\gamma_{\delta; 2(n-p), \frac{p}{\sigma^2}}$ in Lemma~\ref{lem:gamma_delta_Bound}, we obtain an explicit but much relaxed bound as the following.
\begin{cor}\label{cor:bound.settingB_Ver3}
Under setting (B) and when $\e<1$, a sufficient condition for the mechanism $Y(\matX)$ in \eqref{eq:modelB} to achieve $(\e, \delta)$-DP is
\be\label{eq:sigma.bound.new1}
\sigma \ge  \sqrt{\frac{2n-p+\ln(\frac{1}{\delta})}{2(n-p)}} \frac{3\sqrt[4]{p}}{\sqrt{\e}}.
\ee
\end{cor}
The proofs of Corollary~\ref{cor:bound.settingA} and Corollary~\ref{cor:bound.settingB_Ver3} is provided in Appendices~\ref{sec:proof.cor.A} and \ref{sec:proof.cor.B}.

We compare the sufficient condition of DP on additive noise variance in Corollary~\ref{cor:bound.settingB_Ver3} for setting (B), i.e., $\sigma^2 \ge \frac{2n-p+\ln(\frac{1}{\delta})}{2(n-p)} \frac{3 \sqrt{p}}{\e}$, with the sufficient condition in Corollary~\ref{cor:bound.settingA} for setting (A), $\sigma^2 \ge \frac{1.7 \ln(\frac{1}{\delta})}{\e^2}$, under $\delta < 0.05$ and $\e < 1$. Because the number $n$ of samples is typically much larger than the number $p$ of features in each sample in practice, we may simplify the lower bound for setting (B) as $\sigma^2 \ge (1+\frac{\ln(\frac{1}{\delta})}{2n}) \frac{3 \sqrt{p}}{\e}$.
First, consider the bounds in terms of $\e$ only, under any given configuration of other parameters. With matrix masking, the lower bound on the additive noise variance $\sigma^2$ improves by a factor $\e$, from $O(\frac{1}{\e^2})$ of setting (A) to $O(\frac{1}{\e})$ of setting (B). Second, consider the bounds in terms of $\delta$ and $\e$, under any given configuration of $n$ and $p$. The lower bound on noise variance improves from $O(\frac{ln(\frac{1}{\delta})}{\e^2})$ to $O(\frac{\sqrt{ln(\frac{1}{\delta})}}{\e})$.


\section{Discussion on the Implication of the Analysis Results}\label{sec:discuss}

To further understand how much advantage is provided by the matrix masking over the simple additive noise only, we compare the bounds \eqref{eq:sigma.bound.new} versus \eqref{eq:DP-boundA.suf} for various $\e$, $\delta$, $n$ and $p$ values in Table~\ref{bound.table}.

\begin{table}[h]
\caption{Comparison of $\sigma$ bounds \eqref{eq:DP-boundA.suf} for setting (A) versus \eqref{eq:sigma.bound.new} for setting (B).}
\label{bound.table}
\centering
\begin{tabular}{@{}llrrrrrr@{}}
  \hline
  $\e$ & $\delta$ & $p$ & $n$
 & \multicolumn{1}{>{\arraybackslash}m{1.9cm}}{ Setting (A) \newline necessary \eqref{eq:DP-boundA.nec} }
 & \multicolumn{1}{>{\arraybackslash}m{1.9cm}}{ Setting (A) \newline sufficient \eqref{eq:DP-boundA.suf} }
 & \multicolumn{1}{>{\arraybackslash}m{1.9cm}}{ Setting (B) \newline sufficient \eqref{eq:sigma.bound.new}}
 & \multicolumn{1}{>{\arraybackslash}m{1.5cm}}{ Ratio of \newline \eqref{eq:DP-boundA.suf}/\eqref{eq:sigma.bound.new}}\\
  \hline
0.100 & 0.010 & 1 & 100 & 23.3 & 25.4 & 6.9 & 4 \\
      &       &   & 10000 & 23.3 & 25.4 & 6.4 & 4 \\
      &       & 5 & 100 & 23.3 & 25.4 & 9.5 & 3 \\
      &       &   & 10000 & 23.3 & 25.4 & 8.9 & 3 \\
      &       & 20 & 100 & 23.3 & 25.4 & 13.1 & 2 \\
      &       &   & 10000 & 23.3 & 25.4 & 12.1 & 2 \\
      & 0.001 & 1 & 100 & 30.9 & 32.5 & 7.1 & 5 \\
      &       &   & 10000 & 30.9 & 32.5 & 6.4 & 5 \\
      &       & 5 & 100 & 30.9 & 32.5 & 9.8 & 3 \\
      &       &   & 10000 & 30.9 & 32.5 & 8.9 & 4 \\
      &       & 20 & 100 & 30.9 & 32.5 & 13.5 & 2 \\
      &       &   & 10000 & 30.9 & 32.5 & 12.1 & 3 \\
0.010 & 0.010 & 1 & 100 & 232.6 & 254.1 & 21.8 & 12 \\
      &       &   & 10000 & 232.6 & 254.1 & 20.2 & 13 \\
      &       & 5 & 100 & 232.6 & 254.1 & 30.2 & 8 \\
      &       &   & 10000 & 232.6 & 254.1 & 28.0 & 9 \\
      &       & 20 & 100 & 232.6 & 254.1 & 41.5 & 6 \\
      &       &   & 10000 & 232.6 & 254.1 & 38.3 & 7 \\
      & 0.001 & 1 & 100 & 309.0 & 325.2 & 22.4 & 15 \\
      &       &   & 10000 & 309.0 & 325.2 & 20.2 & 16 \\
      &       & 5 & 100 & 309.0 & 325.2 & 31.0 & 10 \\
      &       &   & 10000 & 309.0 & 325.2 & 28.1 & 12 \\
      &       & 20 & 100 & 309.0 & 325.2 & 42.7 & 8 \\
      &       &   & 10000 & 309.0 & 325.2 & 38.4 & 8 \\
0.001 & 0.010 & 1 & 100 & 2326.3 & 2541.3 & 68.9 & 37 \\
      &       &   & 10000 & 2326.3 & 2541.3 & 63.8 & 40 \\
      &       & 5 & 100 & 2326.3 & 2541.3 & 95.4 & 27 \\
      &       &   & 10000 & 2326.3 & 2541.3 & 88.5 & 29 \\
      &       & 20 & 100 & 2326.3 & 2541.3 & 131.1 & 19 \\
      &       &   & 10000 & 2326.3 & 2541.3 & 121.0 & 21 \\
      & 0.001 & 1 & 100 & 3090.2 & 3252.0 & 70.8 & 46 \\
      &       &   & 10000 & 3090.2 & 3252.0 & 64.0 & 51 \\
      &       & 5 & 100 & 3090.2 & 3252.0 & 98.0 & 33 \\
      &       &   & 10000 & 3090.2 & 3252.0 & 88.8 & 37 \\
      &       & 20 & 100 & 3090.2 & 3252.0 & 134.9 & 24 \\
      &       &   & 10000 & 3090.2 & 3252.0 & 121.4 & 27 \\
   \hline
\end{tabular}
\end{table}

Observing Table~\ref{bound.table}, as expected, the advantage of applying the matrix masking grows as the values of $\e$ and $\delta$ decrease, or as the sample size $n$ increases.
As shown in the last column of the table, when $\e=0.1$, the required noise standard deviation $\sigma$ is reduced by two to five times; while when $\e=0.001$, the required $\sigma$ is reduced from twenty to fifty times. The reduction is also bigger for the smaller $\delta=0.001$ value versus $\delta=0.01$. Thus for applications with strict privacy requirements, the matrix masking drastically reduces the magnitude of additive noise needed, leading to much more accurate result at the same sample size for data analysis methods  that are supported by matrix masking (such as linear models) ~\citep{perturbation,AX,wu2017new}, where the pseudo data matrix after orthogonal transformation has the same sufficient statistics as the raw data matrix before the transformation.

We do note that the advantage of the matrix masking is reduced as the number of features $p$ in data set increases in the above table, since the bound \eqref{eq:sigma.bound.new} increases as $p$ increases (this can be observed clearly from \eqref{eq:sigma.bound.new1}) while \eqref{eq:DP-boundA.suf} is invariant to the changes in $p$. However, it is not clear whether the dependence of \eqref{eq:sigma.bound.new} on $p$ is intrinsic to the problem or is due to proof techniques we used. A future research topic is to study whether the bound \eqref{eq:sigma.bound.new} for the matrix masking settings can be further improved. In any case, the current bound still demonstrates the advantage of matrix masking in the case of big data $n >> p$. For example, when $p=20$, for the case of $\e=0.001$ and $\delta=0.001$, the matrix masking reduces the needed noise magnitude by around $27$ times when $n \ge 10000$.

For the matrix masking settings (B), as we have pointed out above,  \eqref{eq:sigma.bound.new} provides a lower bound on $\sigma$ which may be further improved in the future. In contrast, the bound \eqref{eq:DP-boundA.suf} for setting (A) without matrix masking cannot be improved much better since the necessary condition \eqref{eq:DP-boundA.nec} is very close. As shown in Table~\ref{bound.table} column 5 versus column 6, the necessary condition \eqref{eq:DP-boundA.nec} is less than $10\%$ away from the sufficient condition \eqref{eq:DP-boundA.suf} in all cases.

We have analyzed the situation of adding noise then apply matrix masking. Another way of combining the matrix masking and noise addition is as the following.

(C) Apply matrix masking first then add noise: release
\be\label{eq:modelC}
\matY =  \matA \matX + \matC \qquad \mbox{ where } \matC \sim NI_{n \times p}(0, \sigma^2).
\ee

For this setting (C), we can show that the density of $\matY$ conditional on $\matX$ is the same as the density under setting (B), as shown in Lemma~\ref{lem:DenC} of section~\ref{sec:DenC}. Thus the $\sigma$ needed to achieve $(\e, \delta)$-DP is the same in setting (B) and in setting (C). We focus on the setting (B) analysis only in this paper since the local noise addition combining with the triple-matrix masking scheme can provide end-to-end privacy protection in setting (B) with a much lower added noise magnitude than needed in setting (A). More accurate statistical analysis can be conducted on the pseudo data set with lower added noise.

Finally, when $n$ is big, we may implement the masking scheme with a random block-diagonal matrix to reduce the computational complexity. That is, we take
$$
\matA = \left(\ba{cccc} \matA_1 & \mat0   & \cdots & \mat0 \\
                        \mat0   & \matA_2 & \cdots & \mat0 \\
                        \vdots  & \vdots  & \ddots & \vdots \\
                        \mat0   & \mat0   & \cdots & \matA_k \\    \ea  \right)
$$
where $\matA_1,...,\matA_k$ are respectively orthogonal matrices of sizes $n_1, ..., n_k$ with $n=n_1+...+n_k$.  Let $n_s$ denotes the smallest size among $n_1, ..., n_k$. We assumes that $n_s > p$. Then the block-diagonal matrix mechanism still achieves $(\e, \delta)$-DP when $n$ is replaced by $n_s$ in condition \eqref{eq:sigma.bound.new}.

\section{Proofs of Main Theorems}\label{sec:proof}

\subsection{Proof of Lemma~\ref{lem:DenRatio}}
\begin{proof}[Proof of Lemma~\ref{lem:DenRatio}]

Under setting (A), since $\matC = \matY - \matX$, $p_{\matX}(\maty) = p_{\matC}(\maty - \matX)$ where the density $p_{\matC}$ is the multivariate Gaussian distribution density. That is,
$$
p_{\matX}(\maty) = (\frac{1}{\sqrt{2\pi}\sigma})^{np} e^{-\frac{\|\maty - \matX  \|^2}{2 \sigma^2}}.
$$
Notice that for a $n \times p$ matrix $\matX$, $\|\matX  \|^2= tr(\matX^T\matX) = tr(\matX \matX^T)$. Hence
$$\|\maty - \matX  \|^2 = tr[(\maty - \matX)(\maty - \matX)^T]= \|\maty\|^2 + \|\matX\|^2 - 2tr(\maty\matX^T). $$
Thus the density becomes
\be\label{eq:denA}
p_{\matX}(\maty)  = (\frac{1}{\sqrt{2\pi}\sigma})^{np} e^{-\frac{\|\maty\|^2 + \|\matX\|^2 - 2tr(\maty\matX^T)}{2 \sigma^2}} = (\frac{1}{\sqrt{2\pi}\sigma})^{np} e^{-\frac{\|\maty\|^2}{2 \sigma^2}} e^{-\frac{\|\matX  \|^2}{2 \sigma^2}} e^{\frac{tr(\maty\matX^T)}{ \sigma^2}}. 
\ee
So the ratio of densities under $\matX$ and $\matX'$ is
\be
\frac {p_{\matX}(\maty)}{p_{\matX'}(\maty)}   =\frac {e^{-\frac{\|\matX  \|^2}{2 \sigma^2}} e^{\frac{tr(\maty\matX^T)}{ \sigma^2}}}{e^{-\frac{\|\matX'  \|^2}{2 \sigma^2}} e^{\frac{tr(\maty(\matX')^T)}{ \sigma^2}}}
\qquad = e^{\frac{\|\matX'\|^2 - \|\matX\|^2 }{2 \sigma^2}} \frac {e^{\frac{tr(\maty\matX^T)}{ \sigma^2}}}{e^{\frac{tr[\maty(\matX')^T] }{\sigma^2}}}.
\ee
This expression is indeed \eqref{eq:denRatioA0}.

Under setting (B) $\matY =  \matA (\matX + \matC)$, the joint density of $(\matY, \matA)$ is
$$p_{\matX}(\maty, \matA) = (\frac{1}{\sqrt{2\pi}\sigma})^{np} e^{-\frac{\|\matA^T \maty - \matX\|^2}{2 \sigma^2}} = (\frac{1}{\sqrt{2\pi}\sigma})^{np} e^{-\frac{\|\matA^T\maty\|^2}{2 \sigma^2}} e^{-\frac{\|\matX  \|^2}{2 \sigma^2}} e^{\frac{tr(\matA^T\maty\matX^T)}{ \sigma^2}}.$$
Since $\matA$ is an orthogonal matrix, $\|\matA^T\maty\| = \|\maty\|$.

Integrating the joint density $p_{\matX}(\maty, \matA)$ over the uniform distribution $\mu(\cdot)$ on the group $\mathcal{O}_{n \times n}$, we get the density  
\be\label{eq:DenB}
\ba{cl}
p_{\matX}(\maty) & = \int_{\matA \in \mathcal{O}_{n \times n}} (\frac{1}{\sqrt{2\pi}\sigma})^{np} e^{-\frac{\|\maty\|^2}{2 \sigma^2}} e^{-\frac{\|\matX  \|^2}{2 \sigma^2}} e^{\frac{tr(\matA^T\maty\matX^T)}{ \sigma^2}} d \mu(\matA) \\
& = (\frac{1}{\sqrt{2\pi}\sigma})^{np} e^{-\frac{\|\maty\|^2}{2 \sigma^2}} e^{-\frac{\|\matX  \|^2}{2 \sigma^2}} \int_{\matA \in \mathcal{O}_{n \times n}}  e^{\frac{tr(\matA^T\maty\matX^T)}{ \sigma^2}} d \mu(\matA).
\ea
\ee
So the density ratio becomes \eqref{eq:denRatioB0}:
$$
\frac {p_{\matX}(\maty)}{p_{\matX'}(\maty)} = e^{\frac{\|\matX'\|^2 - \|\matX\|^2 }{2 \sigma^2}} \frac{\int_{\matA \in \mathcal{O}_{n \times n}} e^{\frac{ tr(\matA^T\maty\matX^T)}{ \sigma^2}} d \mu(\matA)}{\int_{\matA \in \mathcal{O}_{n \times n}} e^{\frac{tr[\matA^T\maty(\matX')^T]}{\sigma^2}} d \mu(\matA)}.
$$
\end{proof}

\subsection{Proof of Theorem~\ref{thm:bound.settingA} in setting (A) unmasked release}
\begin{proof}[Proof of Theorem~\ref{thm:bound.settingA}]
We derive the bounds by solving the set $\mathcal{S}_{\matX,\matX'}= \{\maty: p_{\matX}(\maty) > e^\e p_{\matX'}(\maty)\}$. Under setting (A), the density ratio is given by \eqref{eq:denRatioA0}
$$
\frac {p_{\matX}(\maty)}{p_{\matX'}(\maty)} = e^{\frac{\|\matX'\|^2 - \|\matX\|^2 }{2 \sigma^2}} \frac {e^{\frac{tr(\maty\matX^T)}{ \sigma^2}}}{e^{\frac{tr[\maty(\matX')^T] }{\sigma^2}}}.
$$

We want to further simplify \eqref{eq:denRatioA0}. Recall that we denote $\matDelta = \matX' - \matX$ so that $\matX' = \matX + \matDelta$. Similar to \eqref{eq:rows}, we further denote the $i$-th row of $\matDelta$, $\maty$ and $\matC$ respective as $\matDelta_i$, $\maty_i$ and $\matC_i$.
Then $tr[\maty(\matX')^T]=tr(\maty\matX^T)+tr(\maty \matDelta^T ) = tr(\maty\matX^T)+\sum_{k=1}^{n}  \maty_k \matDelta_k^T$.
Notice that here $\|\matDelta_2\| = ... =\|\matDelta_n\| =0$ since $\matX$ and $\matX'$ only differ in the first row. Hence $tr[\maty(\matX')^T]=tr(\maty\matX^T)+ \maty_1 \matDelta_1^T $ and
$$
\|\matX'\|^2 - \|\matX\|^2 = \|\matX + \matDelta\|^2 - \|\matX\|^2 = \| \matDelta\|^2 + 2tr(\matX \matDelta^T) = \| \matDelta_1\|^2 + 2  \matX_1 \matDelta_1^T.
$$
Thus \eqref{eq:denRatioA0} is simplified to
\be\label{eq:denRatioA1}
\frac {p_{\matX}(\maty)}{p_{\matX'}(\maty)}  \ =e^{\frac{\|\matX'\|^2 - \|\matX\|^2 }{2 \sigma^2}} \frac {e^{\frac{tr(\maty\matX^T)}{ \sigma^2}}}{e^{\frac{tr[\maty(\matX')^T] }{\sigma^2}}} \ = e^{\frac{\| \matDelta_1\|^2 + 2 \matX_1 \matDelta_1^T}{2 \sigma^2}} e^{\frac{ -  \maty_1 \matDelta_1^T }{\sigma^2}} \ \ = e^{\frac{\| \matDelta_1\|^2 - 2 (\maty_1-\matX_1) \matDelta_1^T}{2 \sigma^2}}.
\ee

Hence the density ratio exceeding $e^\e$ is equivalent to
$$
\| \matDelta_1\|^2   - 2 (\maty_1-\matX_1)\matDelta_1^T  > 2\sigma^2 \e \qquad \Leftrightarrow \qquad - 2 (\maty_1-\matX_1)\matDelta_1^T > 2\sigma^2 \e - \| \matDelta_1\|^2.
$$
Hence $\mathcal{S}_{\matX,\matX'} = \mathcal{S}_{\matX,\matX + \matDelta} = \{\maty: - (\maty_1-\matX_1)\matDelta_1^T >  \sigma^2\e - \frac{\| \matDelta_1\|^2}{2} \}$.
Notice that $(\matY_1-\matX_1) = \matC_1$ follows a $p$-dimensional Gaussian distribution with zero mean and variance $\sigma^2 I_p$ where $I_p$ denotes the $p \times p$ identity matrix. Hence $-(\matY_1-\matX_1)\matDelta_1^T$ follows the Gaussian distribution with zero mean and variance $\sigma^2 \| \matDelta_1\|^2$. Thus $-\frac{(\matY_1-\matX_1)\matDelta_1^T}{\sigma \| \matDelta_1\|}$ follows the standard Gaussian distribution $N(0,1)$ whose upper $\delta$-quantile is $\bar \gamma_{\delta}$.
Hence condition~\eqref{eq:DP-bound0} is equivalent to
\be\label{eq:DP-bound0A}
\ba{ccl}
 P_\matX[Y \in \mathcal{S}_{\matX,\matX + \matDelta}] \le \delta \qquad &
\Leftrightarrow &  P_\matX[- (\matY_1-\matX_1)\matDelta_1^T >  \sigma^2\e - \frac{\| \matDelta_1 \|^2}{2}] \le \delta \\
& \Leftrightarrow & P_\matX[ -\frac{(\matY_1-\matX_1)\matDelta_1^T}{\sigma \| \matDelta_1\|} >  \frac{\sigma\e}{\| \matDelta_1\|} - \frac{\| \matDelta_1\|}{2\sigma}] \le \delta \\
& \Leftrightarrow &  \bar \gamma_{\delta} \ \ \le  \frac{\sigma\e}{\| \matDelta_1\|} - \frac{\| \matDelta_1\|}{2\sigma} \qquad = \frac{\sigma\e}{\| \matDelta\|} - \frac{\| \matDelta\|}{2\sigma}.
\ea
\ee
Here the last line comes from the fact that $ P_\matX[ -\frac{(\matY_1-\matX_1)\matDelta_1^T}{\sigma \| \matDelta_1\|} > \bar \gamma_{\delta}]= \delta$.

Since $\bar \gamma_{\delta} \le  \frac{\sigma\e}{\| \matDelta\|} - \frac{\| \matDelta\|}{2\sigma} \le \frac{\sigma\e}{\| \matDelta\|}$ always, a necessary condition for \eqref{eq:DP-bound0A} to hold is
$$
\bar \gamma_{\delta} \le  \frac{\sigma\e}{\| \matDelta \|} \qquad \Leftrightarrow \qquad \sigma \ge \frac{\| \matDelta\| \bar \gamma_{\delta}}{\e}.
$$
Thus we finished the proof for \eqref{eq:DP-bound0A.nec}.

Now $\sigma \ge \frac{\| \matDelta\| \bar \gamma_{\delta}}{\e}$ is necessary for \eqref{eq:DP-bound0} to hold for a given pair of neighboring $\matX$ and $\matX'$ that differs by $\|\matDelta\| \le 1$. Since there exists pairs of $\matX$ and $\matX'$ with $\|\matDelta\| = 1$, for \eqref{eq:DP-bound0} to hold for {\it every} pair of neighboring $\matX$ and $\matX'$, a necessary condition is that
$$
\sigma \ge \frac{  \bar \gamma_{\delta}}{\e},
$$
i.e., condition \eqref{eq:DP-boundA.nec} holds.

What remains to be proven is the sufficient condition. Particularly, we want to prove that \eqref{eq:DP-bound0A.suf} $\sigma \ge \frac{\| \matDelta\| \bar \gamma_{\delta}}{\e}(1+\frac{1}{2 \bar \gamma_{\delta}^2})$ is sufficient for \eqref{eq:DP-bound0A}. Since $\frac{\sigma\e}{\| \matDelta\|} - \frac{\| \matDelta\|}{2\sigma}$ is an increasing function of $\sigma$, when $\sigma \ge \frac{ \| \matDelta\| \bar \gamma_{\delta}}{\e}(1+\frac{1}{2 \bar \gamma_{\delta}^2})$, we have
$$
\ba{cl}
\frac{\sigma\e}{\| \matDelta\|} - \frac{\| \matDelta\|}{2\sigma} & \ge \frac{\e}{\| \matDelta\|} \frac{ \| \matDelta\| \bar \gamma_{\delta}}{\e}(1+\frac{1}{2 \bar \gamma_{\delta}^2}) - \frac{\| \matDelta\|}{2\frac{ \| \matDelta\| \bar \gamma_{\delta}}{\e}(1+\frac{1}{2 \bar \gamma_{\delta}^2})}
\\
& = \bar \gamma_{\delta} (1+\frac{1}{2 \bar \gamma_{\delta}^2}) - \frac{\e}{2\bar \gamma_{\delta}(1+\frac{1}{2 \bar \gamma_{\delta}^2})} \\
& = \bar \gamma_{\delta} +\frac{1}{2 \bar \gamma_{\delta}} - \frac{1}{2\bar \gamma_{\delta} } \frac{\e}{1+\frac{1}{2 \bar \gamma_{\delta}^2}} \\
& = \bar \gamma_{\delta} +\frac{1}{2 \bar \gamma_{\delta}} (1 - \frac{\e}{1+\frac{1}{2 \bar \gamma_{\delta}^2}}) \\
& \ge \bar \gamma_{\delta},
 \ea
$$
where the last line is due to $\e \le 1$. This is \eqref{eq:DP-bound0A}.
Thus \eqref{eq:DP-bound0A.suf} is sufficient for \eqref{eq:DP-bound0A} which then implies \eqref{eq:DP-bound0}.

When \eqref{eq:DP-boundA.suf} $\sigma^2 \ge \frac{ \bar \gamma_{\delta}^2}{\e^2}(1+\frac{1}{2 \bar \gamma_{\delta}^2})^2$ holds and $\|\matDelta\| \le 1$, then \eqref{eq:DP-bound0A.suf} $\sigma \ge \frac{\| \matDelta\| \bar \gamma_{\delta}}{\e}(1+\frac{1}{2 \bar \gamma_{\delta}^2})$ holds also. Thus \eqref{eq:DP-boundA.suf} is sufficient for \eqref{eq:DP-bound0} to hold for every pair of neighboring $\matX$ and $\matX'$ with $\|\matDelta\| \le 1$.

\end{proof}

\subsection{Proof of Theorem~\ref{thm:bound.settingB_Ver3} in setting (B)}
\begin{proof}[Proof of Theorem~\ref{thm:bound.settingB_Ver3}]

The main idea of this proof involves finding a high probability set $\mathcal{W}$ on the joint space of $(\matY, \matA)$ satisfying the following property: For any set $\mathcal{S}$ on the space of $\matY$,
\be\label{eq:bound.prob.SW}
P_{\matX} [\matY \in \mathcal{S} \mbox{ and } (\matY, \matA) \in \mathcal{W}] \le e^\e P_{\matX'} [\matY \in \mathcal{S} \mbox{ and } (\matY, \matA) \in \mathcal{W}].
\ee
Then assuming that the set $\mathcal{W}$ has a high probability occurring when raw data set is $\matX$,
\be\label{eq:W.probHigh}
P_{\matX} [(\matY, \matA) \in \mathcal{W}] \ge 1 - \delta,
\ee
we can show that $\matY$ in setting (B) achieves $(\e, \delta)$-DP because
$$
\ba{cl}
P_{\matX} [\matY \in \mathcal{S}] & = P_{\matX} [\matY \in \mathcal{S} \mbox{ and } (\matY, \matA) \in \mathcal{W}] + P_{\matX} [\matY \in \mathcal{S} \mbox{ and } (\matY, \matA) \notin \mathcal{W}] \\
\mbox{by \eqref{eq:bound.prob.SW} } & \le e^\e P_{\matX'} [\matY \in \mathcal{S} \mbox{ and } (\matY, \matA) \in \mathcal{W}] + P_{\matX} [\matY \in \mathcal{S} \mbox{ and } (\matY, \matA) \notin \mathcal{W}] \\
& \le e^\e P_{\matX'} [\matY \in \mathcal{S} ] + P_{\matX} [(\matY, \matA) \notin \mathcal{W}] \\
\mbox{by \eqref{eq:W.probHigh} } & \le e^\e P_{\matX'} [\matY \in \mathcal{S}] + \delta.
\ea
$$

Hence the proof is complete when we specify a set $\mathcal{W}$ that satisfies \eqref{eq:bound.prob.SW}, and also show that condition \eqref{eq:sigma.bound.new} ensures \eqref{eq:W.probHigh} holds. To describe this set clearly, we first conduct a coordinate system change to focus our attention to a $p-$dimensional subspace where the common part of raw data sets $X$ and $X'$ concentrates on. The integral over this subspace makes the ratio between densities for $Y$ under $\matX$ and $\matX'$ becomes bounded by $e^\e$ much easier than in setting (A).

Since $\matX$ and $\matX'$ only differ in the first row, their last $n-1$ rows are the same and belongs to a $p-$dimnesional linear subspace. For a more convenient presentation of this subspace, we apply a QR decomposition on the submatrix consisting of all the common $n-1$ rows in $\matX$ and $\matX'$, so that
$$
\matX  = \left( \ba{cc} 1 & \mat0 \\ \mat0 & \matQ \ea \right)  \left( \ba{c} \matX_1 \\ \matR \ea \right) = \overline{\matQ} \ \overline{\matX}, \qquad \matX'  = \left( \ba{cc} 1 & \mat0 \\ \mat0 & \matQ \ea \right)  \left( \ba{c} \matX_1' \\ \matR \ea \right) = \overline{\matQ} \ \overline{\matX}',
$$
where $\matQ$ is a $(n-1) \times (n-1)$ orthogonal matrix and $\matR$ is an $(n-1) \times p$ upper triangular matrix. Hence $\overline{\matQ}$ is a $n \times n$ orthogonal matrix.
$\overline{\matX}$ and $\overline{\matX}'$ only differ in the first row, which correspondingly are the first rows of $\matX$ and $\matX'$ respectively: $\overline{\matX}_1 = \matX_1$ and $\overline{\matX}_1' = \matX_1'$. The rest rows of $\overline{\matX}$ and $\overline{\matX}'$ are common for both matrices: only the $2$nd to $(p+1)$-th rows $\overline{\matX}_2, ..., \overline{\matX}_{p+1}$ are non-zero, and last $n-p-1$ rows are all zero vectors. That is, after the linear transformation $\overline{\matQ}$, the common part of $\overline{\matX}$ and $\overline{\matX}'$ lie on the $p-$dimensional subspace spanned by $\overline{\matX}_2, ..., \overline{\matX}_{p+1}$.

We apply the linear transformation $\overline{\matQ}$ on $\matA$, that is, we denote $\matB = (\matA \overline{\matQ})^T$ and consider the problem in terms of $\matB$. We further decompose the matrices $\overline{\matX}$, $\overline{\matX}'$ and $\matB$ as
$$
\overline{\matX} = \left( \ba{c} \matX_1 \\ \overline{\matX}_+ \\ \overline{\matX}_0 \ea \right), \qquad \overline{\matX}' = \left( \ba{c} \matX_1' \\ \overline{\matX}_+ \\ \overline{\matX}_0 \ea \right), \qquad {\matB} = \left( \ba{c} \matB_1 \\ {\matB}_+ \\ {\matB}_0 \ea \right),
$$
where $\overline{\matX}_+$ is the $p \times n$ matrix consisting the non-zero rows $\overline{\matX}_2, ..., \overline{\matX}_{p+1}$ of the raw data set and ${\matB}_+$ consists of the corresponding rows from the masking matrix $\matB$;  $\overline{\matX}_0$ is the $(n-p-1) \times n$ zero matrix and ${\matB}_0$ consists of the corresponding rows from the masking matrix $\matB$.

Since the density of uniform distribution on $\mathcal{O}_{n \times n}$, the group of orthogonal ${n \times n}$ matrices, is invariant to multiplication of an orthogonal matrix $\overline{\matQ}$, $\matB$ also follows the uniform distribution on $\mathcal{O}_{n \times n}$. We further denote
$$
\qquad {\matB}_+^\bot = \left( \ba{c} \matB_1 \\ {\matB}_0 \ea \right).
$$
Let $\mathcal{O}_{p \times n}$ denote the set of ${p \times n}$ matrices consisting of $p$ orthonormal row vectors, and $\mathcal{O}_{(n-p) \times n} (\bot {\matB}_+)$ denote the set of ${(n-p) \times n}$ matrices consisting of $n-p$ orthonormal row vectors which are also orthogonal to rows of ${\matB}_+$.
Since $\matB$ follows the uniform distribution on $\mathcal{O}_{n \times n}$, we have the following two facts:
(1) $\matB_+$ follows the uniform distribution on $\mathcal{O}_{p \times n}$; (2) given a $\matB_+$, ${\matB}_+^\bot$ follows the uniform distribution on $\mathcal{O}_{(n-p) \times n} (\bot {\matB}_+)$. In the following, we always use $\mu$ to denote the measure for uniform distribution on each corresponding group: either $\mathcal{O}_{n \times n}$ or $\mathcal{O}_{p \times n}$ or $\mathcal{O}_{(n-p) \times n} (\bot {\matB}_+)$.

To choose a set $\mathcal{W}$ that satisfies \eqref{eq:bound.prob.SW}, we observe how the densities $p_{\matX}(\maty)$ and $p_{\matX'}(\maty)$ differ. Particularly we want to concentrate on the differences in ${\matB}_+^\bot$.
Recall the density of masked data, when the raw data is $\matX$, is given by
$$
p_\matX(\maty) = \int_{\matA \in \mathcal{O}_{n \times n}} (\frac{1}{\sqrt{2\pi}\sigma})^{np}e^{-\frac{\|\maty\|^2 + \|\matX\|^2 - 2tr(\matA^T\maty {\matX}^T)}{2 \sigma^2}} d \mu(\matA).
$$
Since $\matB^T \overline{\matX} = (\matA \overline{\matQ}) (\overline{\matQ}^T \matX) = \matA \matX$ we have
$$tr(\matA^T \maty\matX^T) = tr[\maty(\matA \matX)^T] = tr[\maty(\matB^T \overline{\matX})^T] = tr(\maty \overline{\matX}^T \matB) = tr(\matB \maty\overline{\matX}^T)=\sum_{i=1}^{p+1}  \matB_i \maty \overline{\matX}_i^T.$$
Here the last equality is due to the fact that the last $n-p-1$ rows of $\overline{\matX}$ are all zero vectors, thus the trace only involves the sum for the first $p+1$ rows. Using this expression, the density becomes
\be\label{eq:den.X}
\ba{cl}
& p_\matX(\maty) \\
=& \int_{\matB \in \mathcal{O}_{n \times n}} (\frac{1}{\sqrt{2\pi}\sigma})^{np}e^{-\frac{\|\maty\|^2 + \|\matX\|^2 - 2tr(\matB \maty\overline{\matX}^T)}{2 \sigma^2}} d \mu(\matB) \\
=& \int_{\matB \in \mathcal{O}_{n \times n}} (\frac{1}{\sqrt{2\pi}\sigma})^{np}e^{-\frac{\|\maty\|^2 + \|\matX\|^2}{2 \sigma^2}} e^{\frac{\matB_1 \maty \matX_1^T +  \sum_{i=2}^{p+1}  \matB_i \maty \overline{\matX}_i^T}{ \sigma^2}} d \mu(\matB) \\
=& \int_{\matB_+ \in \mathcal{O}_{p \times n}} (\frac{1}{\sqrt{2\pi}\sigma})^{np}e^{-\frac{\|\maty\|^2 + \|\matX\|^2}{2 \sigma^2}} e^{\frac{ \sum_{i=2}^{p+1}  \matB_i \maty \overline{\matX}_i^T}{ \sigma^2}} [\int_{\matB_+^\bot \in \mathcal{O}_{(n-p) \times n} (\bot {\matB}_+)} e^{\frac{ \matB_1 \maty \matX_1^T }{\sigma^2}}  d \mu(\matB_+^\bot)] d \mu(\matB_+) \\
=& \int_{\matB_+ \in \mathcal{O}_{p \times n}} f(y,\overline{\matX}_+,{\matB}_+)
e^{-\frac{\|\matX\|^2}{2 \sigma^2}} [\int_{\matB_+^\bot \in \mathcal{O}_{(n-p) \times n} (\bot {\matB}_+)} e^{\frac{ \matB_1 \maty \matX_1^T }{\sigma^2}}  d \mu(\matB_+^\bot)] d \mu(\matB_+).
\ea
\ee
where we denote $f(y,\overline{\matX}_+,{\matB}_+)= (\frac{1}{\sqrt{2\pi}\sigma})^{np} e^{-\frac{\|\maty\|^2 }{2 \sigma^2}} e^{\frac{\sum_{i=2}^{p+1}  \matB_i \maty \overline{\matX}_i^T}{\sigma^2}}$.

The density $p_{\matX'}(\maty)$, when the raw data is $\matX'$, has a similar expression with $\matX$ and $\matX_1$ replaced by respectively $\matX'$ and $\matX_1'$. The ratio of these two densities is therefore
\be\label{eq:denRatio.B}
\ba{cl}
& \frac {p_{Y(\matX)}(\maty)}{p_{Y(\matX')}(\maty)} \\
= & e^{\frac{\|\matX'\|^2 - \|\matX\|^2 }{2 \sigma^2}} \frac{\int_{\matB_+ \in \mathcal{O}_{p \times n}} f(y,\overline{\matX}_+,{\matB}_+) [\int_{\matB_+^\bot \in \mathcal{O}_{(n-p) \times n} (\bot {\matB}_+)} e^{\frac{ \matB_1 \maty \matX_1^T }{\sigma^2}}  d \mu(\matB_+^\bot)] d \mu(\matB_+)}{\int_{\matB_+ \in \mathcal{O}_{p \times n}} f(y,\overline{\matX}_+,{\matB}_+) [\int_{\matB_+^\bot \in \mathcal{O}_{(n-p) \times n} (\bot {\matB}_+)} e^{\frac{ \matB_1 \maty (\matX_1')^T  }{\sigma^2}}  d \mu(\matB_+^\bot)] d \mu(\matB_+)}.
\ea
\ee

To further simplify the expression of the inner integral, we will use the following Lemma (whose proof is in subsection~\ref{sec:proof.unif}). In the following, when not specified, the vectors such as $\matb$ and $\matv$ are $n$-dimensional row vectors.
\begin{lemma}\label{lem:int_unif}
  Let $\mathfrak{S}$ denote a $q$-dimensional linear subspace of the $n$-dimensional Euclidean space $\mathfrak{R}^n$. Let $\bm{proj}_\mathfrak{S} (\matv)$ denotes the projection of a vector $\matv$ onto the subspace $\mathfrak{S}$.
  Let $\mu_q(\cdot)$ denote the measure for the uniform distribution over the $(q-1)$-dimensional unit sphere (the surface of the $q$-dimensional unit ball) within the subspace $\mathfrak{S}$.
  Then for any function $g(\cdot)$,
  $$
  \int g(\matb \matv^T) d \mu_q(\matb) = \int_{u=-1}^{1} g(\|\bm{proj}_\mathfrak{S} (\matv) \| u) \frac{1}{\bar c_q} (1-u^2)^\frac{q-3}{2} du
  $$
where $\bar c_q = \int_{u=-1}^{1} (1-u^2)^\frac{q-3}{2} du = \frac{\Gamma(\frac{1}{2})\Gamma(\frac{q-1}{2})}{\Gamma(\frac{q}{2})}$. Here $\Gamma(\cdot)$ is the Gamma function.
\end{lemma}

Applying the Lemma~\ref{lem:int_unif} to \eqref{eq:denRatio.B}, we have
\be\label{eq:denRatio.B1}
\frac {p_{Y(\matX)}(\maty)}{p_{Y(\matX')}(\maty)} = e^{\frac{\|\matX'\|^2 - \|\matX\|^2 }{2 \sigma^2}} \frac{\int_{\matB_+ \in \mathcal{O}_{p \times n}} f(y,\overline{\matX}_+,{\matB}_+) G_{n-p}(\frac{\|v( \matX_1 \maty^T, {\matB}_+) \| }{ \sigma^2} ) d \mu(\matB_+)}{\tilde \int_{\matB_+ \in \mathcal{O}_{p \times n}} f(y,\overline{\matX}_+,{\matB}_+) G_{n-p}(\frac{\|v( \matX_1' \maty^T, {\matB}_+) \| }{ \sigma^2} ) d \mu(\matB_+)}.
\ee
where $v( \matb, {\matB}_+) = \bm{proj}_{\mathfrak{S} (\bot {\matB}_+)} ( \matb)$ denotes the projection vector of $\matb$ onto the $(n-p)$-dimensional linear subspace $\mathfrak{S} (\bot {\matB}_+)$ perpendicular to all rows in ${\matB}_+$ and
\be
G_q(t) = \int_{u=-1}^{1} e^{ t u} (1-u^2)^\frac{q-2}{2} du.
\ee

The density ratio in \eqref{eq:denRatio.B1} can be bounded using the following Lemma~\ref{lem:GRatio} on the ratio of function $G_q()$, whose proof is in subsection~\ref{sec:proof.GRatio}.
\begin{lemma}\label{lem:GRatio}
When $t > 0$, the derivative of $G_q(t)$ satisfies $0<G_q'(t)<\frac{t}{q}G_q(t)$. Thus $G_q(t)$ is an increasing function when $t \ge 0$, and
\be\label{eq:Gratio1}
\frac{G_q(t_2)}{G_q(t_1)} \le e^{\frac{|t_1^2-t_2^2|}{2q}} \mbox{\; for any $t_1 >0$ and $t_2 >0$.}
\ee
\end{lemma}

The bound \eqref{eq:Gratio1} together with \eqref{eq:denRatio.B1} indicate that we can choose the following set $\mathcal{W}$ to achieve \eqref{eq:bound.prob.SW}.
\be\label{eq:setW}
\mathcal{W}  =  \{ (\maty, \matB): \frac{| \ \|v( \matX_1 \maty^T, {\matB}_+) \|^2 - \|v( \matX_1' \maty^T, {\matB}_+) \|^2 \ |}{ \sigma^4} \le 2 (n-p) [\e + \frac{ \|\matX\|^2 - \|\matX'\|^2}{2 \sigma^2} ] \}.
\ee

For $(\maty, \matB) \in \mathcal{W}$, \eqref{eq:Gratio1} and \eqref{eq:setW} imply that
\be\label{eq:Gratio.n-p}
\ba{cl}
 G_{n-p}(\frac{\|v( \matX_1 \maty^T, {\matB}_+) \| }{ \sigma^2} )
& \le  G_{n-p}(\frac{\|v( \matX_1' \maty^T, {\matB}_+) \| }{ \sigma^2} ) e^{\frac{| \ \|v( \matX_1 \maty^T, {\matB}_+) \|^2 - \|v( \matX_1' \maty^T, {\matB}_+) \|^2}{2(n-p)}} \\
& \le  G_{n-p}(\frac{\|v( \matX_1' \maty^T, {\matB}_+) \| }{ \sigma^2} ) e^\e e^{\frac{ \|\matX\|^2 - \|\matX'\|^2}{2 \sigma^2}}.
\ea
\ee

Let $\mathcal{W}_\matB$ denotes the set $\{ \maty: (\maty, \matB) \in  \mathcal{W} \}$.
Using \eqref{eq:setW}, we note that the set $\mathcal{W}_\matB$ only depends on ${\matB}_+$ and does not depend on the other part of matrix $\matB$, thus we can further denote $\mathcal{W}_\matB$ as $\mathcal{W}_{{\matB}_+}$.  Therefore,
\be\label{eq:bound.prob.SW.B}
\ba{cl}
& P_{\matX} [\matY \in \mathcal{S} \mbox{ and } (\matY, \matB) \in \mathcal{W}] \\
= & \int_{\matB_+ \in \mathcal{O}_{p \times n}} \int_{\matB_+^\bot \in \mathcal{O}_{(n-p) \times n} (\bot {\matB}_+)} \int_{\maty \in \mathcal{S} \cap \mathcal{W}_\matB } f(y,\overline{\matX}_+,{\matB}_+)
e^{-\frac{\|\matX\|^2}{2 \sigma^2}} e^{\frac{ \matB_1 \maty \matX_1^T }{\sigma^2}}  d \maty  d \mu(\matB_+^\bot) d \mu(\matB_+) \\
= & e^{-\frac{\|\matX\|^2}{2 \sigma^2}} \int_{\matB_+ \in \mathcal{O}_{p \times n}}  \int_{\maty \in \mathcal{S} \cap \mathcal{W}_{{\matB}_+} } f(y,\overline{\matX}_+,{\matB}_+) \int_{\matB_+^\bot \in \mathcal{O}_{(n-p) \times n} (\bot {\matB}_+)}
 e^{\frac{ \matB_1 \maty \matX_1^T }{\sigma^2}}   d \mu(\matB_+^\bot) d \maty d \mu(\matB_+)
\\
= & e^{-\frac{\|\matX\|^2}{2 \sigma^2}} \int_{\matB_+ \in \mathcal{O}_{p \times n}}  \int_{\maty \in \mathcal{S} \cap \mathcal{W}_{{\matB}_+} } f(y,\overline{\matX}_+,{\matB}_+)  G_{n-p}(\frac{\|v( \matX_1 \maty^T, {\matB}_+) \| }{ \sigma^2} ) d \maty d \mu(\matB_+)
\\
\le & e^{-\frac{\|\matX\|^2}{2 \sigma^2}} \int_{\matB_+ \in \mathcal{O}_{p \times n}}  \int_{\maty \in \mathcal{S} \cap \mathcal{W}_{{\matB}_+} } f(y,\overline{\matX}_+,{\matB}_+)  G_{n-p}(\frac{\|v( \matX_1' \maty^T, {\matB}_+) \| }{ \sigma^2} ) e^\e e^{\frac{ \|\matX\|^2 - \|\matX'\|^2}{2 \sigma^2}} d \maty d \mu(\matB_+) \\
= & e^\e e^{-\frac{\|\matX'\|^2}{2 \sigma^2}} \int_{\matB_+ \in \mathcal{O}_{p \times n}}  \int_{\maty \in \mathcal{S} \cap \mathcal{W}_{{\matB}_+} } f(y,\overline{\matX}_+,{\matB}_+)  G_{n-p}(\frac{\|v( \matX_1' \maty^T, {\matB}_+) \| }{ \sigma^2} ) d \maty d \mu(\matB_+) \\
= & e^\e P_{\matX'} [\matY \in \mathcal{S} \mbox{ and } (\matY, \matB) \in \mathcal{W}].
\ea
\ee
where the inequality ``$\le$" comes from \eqref{eq:Gratio.n-p}.

Furthermore, the following lemma states that $\mathcal{W}$ is a high probability set, whose proof is provided next.
\begin{lemma}\label{lem:W.probHigh}
When the condition on the noise bound \eqref{eq:sigma.bound.new} holds,
\be\label{eq:W.probHigh.B}
P_{\matX} [(\matY, \matB) \in \mathcal{W}] \ge 1 - \delta.
\ee
\end{lemma}

Notice that \eqref{eq:bound.prob.SW.B} and \eqref{eq:W.probHigh.B} are in fact respectively \eqref{eq:bound.prob.SW} and \eqref{eq:W.probHigh} but stated in terms of $\matB$ instead of $\matA$. Thus the proof is completed.
\end{proof}

\begin{proof}[Proof of Lemma~\ref{lem:W.probHigh}]
The main idea of the proof is to show that  $| \|v( \matX_1' \maty^T, {\matB}_+) \|^2 - \|v( \matX_1 \maty^T, {\matB}_+) \|^2 |$ is bounded by a random variable following the non-central chi-square distribution, and thus \eqref{eq:sigma.bound.new} ensures that $\mathcal{W}$ described in \eqref{eq:setW} occurs with high probability.

Denote $T= \|v( \matX_1' \matY^T, {\matB}_+) \|^2 - \|v_1( \matX_1 \matY^T, {\matB}_+) \|^2$. We first show that $T$ can be expressed as a sum of quadratic forms. Since $\matB_1$, $\matB_{p+2}$,..., $\matB_{n}$ forms an orthogonal basis for the linear subspace $\mathfrak{S} (\bot {\matB}_+)$, the projection of $\matX_1 \maty^T$ onto this subspace is
\be\label{eq:V1.vec}
v( \matX_1 \matY^T, {\matB}_+) = (\matX_1 \matY^T \matB_1^T) \matB_1 + \sum_{i=p+2}^n (\matX_1 \matY^T \matB_i^T) \matB_i.
\ee
Denote $\matZ_1=\matY^T \matB_1^T$ and $\matZ_i=\matY^T \matB_{p+i}^T$ for $i=2,...,n-p$. Then \eqref{eq:V1.vec} implies that
$$
\|v( \matX_1 \matY^T, {\matB}_+) \|^2 = \|\matX_1 \matY^T \matB_1^T\|^2 + \sum_{i=p+2}^n \|\matX_1 \matY^T \matB_i^T\|^2 = \sum_{i=1}^{n-p} \|\matX_1 \matZ_i\|^2 = \sum_{i=1}^{n-p} \matZ_i^T \matX_1^T \matX_1 \matZ_i.
$$
Therefore
\be\label{eq:T.sum}
T = \sum_{i=1}^{n-p} \matZ_i^T (\matX_1')^T \matX_1' \matZ_i - \sum_{i=1}^{n-p} \matZ_i^T \matX_1^T \matX_1 \matZ_i = \sum_{i=1}^{n-p} \matZ_i^T \matM \matZ_i,
\ee
where $\matM = (\matX_1')^T \matX_1' - \matX_1^T \matX_1$. That is, $T$ is expressed as a sum of quadratic forms $\matZ_i^T \matM \matZ_i$.

Next, we show that each of the term $\matZ_i^T \matM \matZ_i$ is bounded by another quadratic form which follows the non-central chi-square distribution up to a constant multiplier.

Since $\matM = (\matX_1')^T \matX_1' - \matX_1^T \matX_1$ is symmetric and its rank is at most two, it has at most two non-zero singular values denoted as $\lambda_1$ and $\lambda_2$. Hence the singular value decomposition of $\matM$ is
$$
\matM = \matV^T \left( \ba{ccc} \lambda_1 & 0 & \mat0_{1 \times (p-2)} \\  0 & \lambda_2 & \mat0_{1 \times (p-2)} \\ \mat0_{(p-2) \times 1} & \mat0_{(p-2) \times 1} & \mat0_{(p-2) \times (p-2)} \ea \right) \matV
$$
where $\matV$ is an orthogonal matrix. Let $\lambda_{max} = \max(|\lambda_1|, |\lambda_2|)$. The property of the quadratic forms implies that
\be\label{eq:ZMZ.bound}
\ba{cl}
|\matZ_i^T \matM \matZ_i| & = | (\matV \matZ_i)^T \left( \ba{ccc} \lambda_1 & 0 & \mat0_{1 \times (p-2)} \\  0 & \lambda_2 & \mat0_{1 \times (p-2)} \\ \mat0_{(p-2) \times 1} & \mat0_{(p-2) \times 1} & \mat0_{(p-2) \times (p-2)} \ea \right) (\matV \matZ_i) | \\
& \le (\matV \matZ_i)^T \left( \ba{ccc} \lambda_{max} & 0 & \mat0_{1 \times (p-2)} \\  0 & \lambda_{max} & \mat0_{1 \times (p-2)} \\ \mat0_{(p-2) \times 1} & \mat0_{(p-2) \times 1} & \mat0_{(p-2) \times (p-2)} \ea \right) (\matV \matZ_i) \\
& = \lambda_{max} \sigma^2 T_i,
\ea
\ee
for
\be\label{eq:Ti}
T_i = (\frac{\matV \matZ_i}{\sigma})^T \left( \ba{ccc} 1 & 0 & \mat0_{1 \times (p-2)} \\  0 & 1 & \mat0_{1 \times (p-2)} \\ \mat0_{(p-2) \times 1} & \mat0_{(p-2) \times 1} & \mat0_{(p-2) \times (p-2)} \ea \right) (\frac{\matV \matZ_i}{\sigma})
\ee
Then we show that the quadratic forms $T_i$'s are independent non-central chi-square random variables.

Notice that, given $\matB$, $\matV \matZ_i = \matV \matY^T \matB_i^T$ is a linear combination of the normal random variables in the matrix $\matY$. Hence $\matV \matZ_i$ is also a normal random vector. When $i \ne j$, since $\matB_i$ and $\matB_j$ are orthogonal, we have
$$
COV(\matB_{i} \matY, \matB_{j} \matY)= E[(\matY - E \matY )^T \matB_{i}^T, \matB_{j} (\matY- E \matY)]= \mat0_{p \times p}.
$$
Thus $\matZ_i$ is independent of $\matZ_j$ for $i \ne j$, because zero covariance implies independence between two normal random vectors. This then implies that $T_i$ is independent of $T_j$ for $i \ne j$.

Since $\| \matB_i \| =1$,  $\matB_i$ times a column of $\matY$ is a $1$-dimensional normal random variable with variance $\|\matB_i\|^2 \sigma^2= \sigma^2$. Also, columns of $\matY$ are independent of each other, thus the elements in the normal random vector $\matB_i \matY$ are independent of each other. Thus
$$
Var( \matB_i \matY)= \sigma^2 \matI_{p \times p}.
$$
That is, the variance matrix for $\matZ_i$ is $Var( \matZ_i)= \sigma^2 \matI_{p \times p}$. Therefore
$$
Var(\matV \matZ_i)= \matV Var( \matZ_i) \matV^T = \matV \sigma^2 \matI_{p \times p} \matV^T = \sigma^2 \matI_{p \times p},
$$
where the last equality is due to the fact that $\matV$ is an orthogonal matrix.

Therefore, $\frac{\matV \matZ_i}{\sigma}$ is a normal random vector with variance matrix $\matI_{p \times p}$. So its elements are independent random variable each with variance one. From \eqref{eq:Ti}, $T_i$ is the sum of squares of the first two elements in $\frac{\matV \matZ_i}{\sigma}$, hence $T_i$ follows a non-central chi-square distribution with non-central parameter $N_i$ equals to the expectation of the sum of squares of its first two elements. Thus $N_i \le \|E (\frac{\matV \matZ_i}{\sigma})\|^2$.

Now we have shown that $T_i$'s are independent non-central chi-square random variables with non-central parameters $N_i$ and each has degree of freedom $2$, hence $\tilde T = \sum_{i=1}^{n-p} |T_i| $ follows a non-central chi-square distribution with non-central parameter $\tilde N = \sum_{i=1}^{n-p} N_i$ and degree of freedom $\sum_{i=1}^{n-p} 2 = 2 (n-p)$. And we note that $|T| \le \lambda_{max} \sigma^2 \tilde T$ from \eqref{eq:T.sum} and \eqref{eq:Ti}. Since $N_i \le \|E (\frac{\matV \matZ_i}{\sigma})\|^2$, $\tilde N \le \sum_{i=1}^{n-p} \|E (\frac{\matV \matZ_i}{\sigma})\|^2$ which will be calculated in the following.

Since $E(\matY) = \matA \matX = \matB^T \overline{\matX} 
$, we have
$$
E(\matV \matZ_1) = \matV E(\matY^T) \matB_1^T = \matV \overline{\matX}^T \matB \matB_1^T = \matV \overline{\matX}_1^T = \matV \matX_1^T,
$$
and for $i=2,...,n-p$,
$$
E(\matV \matZ_i) = \matV E(\matY^T) \matB_{p+i}^T = \matV \overline{\matX}^T \matB \matB_{p+i}^T = \matV \overline{\matX}_{p+i}^T = \mat0_{p \times 1}.
$$
Therefore,
$$
\|E (\frac{\matV \matZ_i}{\sigma})\|^2 = \left\{
                                             \begin{array}{ll}
                                             \frac{\|\matV \matX_1^T\|^2}{\sigma^2} = \frac{\| \matX_1\|^2}{\sigma^2} , & i=1 \\
                                               0, & i=2, ..., n-p
                                             \end{array}
                                           \right.
$$
Thus we have $\tilde N \le \sum_{i=1}^{n-p} \|E (\frac{\matV \matZ_i}{\sigma})\|^2 = \frac{\| \matX_1\|^2}{\sigma^2}$.

From \eqref{eq:XboundedOne}, $\| \matX_1\|^2 \le p$. Hence we have shown that $|T| \le \lambda_{max} \sigma^2 \tilde T$ with $\tilde T$ follows a non-central chi-square distribution with non-central parameter $\tilde N \le \frac{\| \matX_1\|^2}{\sigma^2} \le \frac{p}{\sigma^2}$ and degree of freedom $2(n-p)$.

We are now ready to show that the condition \eqref{eq:sigma.bound.new} ensures the high probability of set $\mathcal{W}$ occurrence. Using \eqref{eq:sigma.bound.new}, we have
\be\label{eq:prob.bound.1}
\ba{ccl}
\frac{(2 \sqrt{p} +1)\chi_{2(n-p)}^2( \frac{p}{\sigma^2}; \delta )}{2 (n-p) }  + \sqrt{p}  \le \sigma^2 \e & \Rightarrow & \frac{(2 \sqrt{p} +1)\chi_{2(n-p)}^2( \tilde N; \delta )}{2 (n-p) }  + \sqrt{p} \le \sigma^2 \e \\
& \Leftrightarrow & \chi_{2(n-p)}^2( \tilde N; \delta ) \le \frac{2 (n-p)}{2 \sqrt{p} +1} \sigma^2 ( \e - \frac{\sqrt{p}}{\sigma^2}).
\ea
\ee

For the symmetric matrix $\matM = (\matX_1')^T \matX_1' - \matX_1^T \matX_1$, $\lambda_{max}$ is also its spectral norm $\|\matM \|_s$ which is bounded above by the  Frobenius norm $\|\matM \|_F$. That is,
$$
\ba{cl}
& \lambda_{max} = \|\matM \|_s  \\
\le & \|\matM \|_F = \sqrt{tr(\matM^T \matM)} = \sqrt{tr(\matM \matM)} \\
= & \sqrt{tr[(\matX_1')^T \matX_1'(\matX_1')^T \matX_1' - (\matX_1')^T \matX_1' \matX_1^T \matX_1 - \matX_1^T \matX_1 (\matX_1')^T \matX_1' + \matX_1^T \matX_1 \matX_1^T \matX_1]} \\
= & \sqrt{tr[\matX_1' (\matX_1')^T \matX_1'(\matX_1')^T]  - tr[\matX_1 (\matX_1')^T \matX_1' \matX_1^T]  - tr[\matX_1'\matX_1^T \matX_1 (\matX_1')^T]  + tr[\matX_1 \matX_1^T \matX_1 \matX_1^T ]} \\
= & \sqrt{[\matX_1' (\matX_1')^T]^2 - 2 [\matX_1' \matX_1^T]^2  + [\matX_1 \matX_1^T]^2},
\ea
$$
where the last equality involves no trace anymore since $\matX_1 \matX_1^T$, $\matX_1' \matX_1^T$ and $\matX_1' (\matX_1')^T$ are all scalars.

Because $\matX_1' = \matX_1 + \matDelta_1$, this implies that
$$
\ba{cl}
 & \lambda_{max}^2 \\
\le  & [\matX_1' (\matX_1 + \matDelta_1)^T]^2 - 2 (\matX_1' \matX_1^T)^2  + (\matX_1 \matX_1^T)^2 \\
= & (\matX_1' \matX_1^T + \matX_1' \matDelta_1^T)^2 - 2 (\matX_1' \matX_1^T)^2  - 2 (\matX_1' \matDelta_1^T)^2 + 2 (\matX_1' \matDelta_1^T)^2 + (\matX_1 \matX_1^T)^2 \\
= & - (\matX_1' \matX_1^T - \matX_1' \matDelta_1^T)^2   + 2 (\matX_1' \matDelta_1^T)^2 + (\matX_1 \matX_1^T)^2 \\
= & - [\matX_1' (\matX_1^T - \matDelta_1^T)]^2   + 2 [\matX_1' \matDelta_1^T]^2 + (\matX_1 \matX_1^T)^2 \\
= & - [(\matX_1 + \matDelta_1) (\matX_1^T - \matDelta_1^T)]^2   + 2 [(\matX_1 + \matDelta_1) \matDelta_1^T]^2 + (\matX_1 \matX_1^T)^2 \\
= & - (\matX_1 \matX_1^T - \matDelta_1 \matDelta_1^T)^2   + 2 (\matX_1 \matDelta_1^T + \matDelta_1\matDelta_1^T) ^2 + (\matX_1 \matX_1^T)^2 \\
= & - ( \|\matX_1\|^2 - \| \matDelta_1 \|^2)^2   + 2 (\matX_1 \matDelta_1^T + \| \matDelta_1 \|^2)^2 + \| \matX_1\|^4 \\
= & - \|\matX_1\|^4 + 2 \|\matX_1\|^2 \| \matDelta_1 \|^2 - \| \matDelta_1 \|^4   + 2 (\matX_1 \matDelta_1^T)^2 + 4(\matX_1 \matDelta_1^T)\| \matDelta_1 \|^2 + 2\| \matDelta_1 \|^4 + \| \matX_1\|^4 \\
= & 2 \|\matX_1\|^2 \| \matDelta_1 \|^2 + 2 (\matX_1 \matDelta_1^T)^2 + 4(\matX_1 \matDelta_1^T)\| \matDelta_1 \|^2 + \| \matDelta_1 \|^4 \\
\le & 2 \|\matX_1\|^2 \| \matDelta_1 \|^2 + 2 \|\matX_1\|^2 \| \matDelta_1 \|^2 + 4\|\matX_1\| \| \matDelta_1 \| \| \matDelta_1 \|^2 + \| \matDelta_1 \|^4 \\
\le & 4 \|\matX_1\|^2 + 4\|\matX_1\| + 1 \qquad \qquad \mbox{(since $\| \matDelta_1 \| \le 1$)} \\
= & (2\|\matX_1\| + 1)^2 \qquad \le  (2 \sqrt{p} +1)^2.
\ea
$$
Thus
$$
\lambda_{max} \le  2 \sqrt{p} +1.
$$
Thus \eqref{eq:prob.bound.1} implies that
$$
\chi_{2(n-p)}^2( \tilde N; \delta ) \le \frac{2 (n-p)}{\lambda_{max}} \sigma^2 ( \e - \frac{\sqrt{p}}{\sigma^2}).
$$
Since $P_\matX[\tilde T \le \chi_{2(n-p)}^2( \tilde N; \delta ) ]=1 - \delta$,
$$
P_\matX[\tilde T \le \frac{2 (n-p)}{\lambda_{max}} \sigma^2 ( \e - \frac{\sqrt{p}}{\sigma^2}) ]\ge 1 - \delta.
$$
Since $|T| \le \lambda_{max} \sigma^2 \tilde T$, we have
\be\label{eq:prob.bound.2}
P_\matX[\frac{| T|}{ \sigma^4} \le 2 (n-p) ( \e - \frac{\sqrt{p}}{\sigma^2}) ] \ge P_\matX[\frac{\lambda_{max} \tilde T}{ \sigma^2}  \le 2 (n-p) ( \e - \frac{\sqrt{p}}{\sigma^2}) ] \ge 1 - \delta.
\ee

Finally, since $|\|\matX_1\|^2 - \|\matX_1'\|^2| = (\|\matX_1\| + \|\matX_1'\|) \ | \ \|\matX_1\| - \|\matX_1'\| \ | \le (\sqrt{p}+\sqrt{p}) \|\matX_1 - \matX_1'\| = 2 \sqrt{p} \| \Delta_1 \| \le 2 \sqrt{p}$, $\|\matX_1\|^2 - \|\matX_1'\|^2 \ge - 2 \sqrt{p}$. Hence \eqref{eq:prob.bound.2} implies that
$$
P_\matX[\frac{| T|}{ \sigma^4} \le 2 (n-p) ( \e + \frac{\|\matX_1\|^2 - \|\matX_1'\|^2}{2\sigma^2}) ] \ge 1 - \delta.
$$
Using \eqref{eq:setW}, we can see that this is \eqref{eq:W.probHigh.B}.

\end{proof}

\begin{appendix}
\section{Proofs of technical results}\label{sec:appendix}

\subsection{Proof of Lemma~\ref{lem:BC<A}}\label{sec:proof.BC<A}

\begin{proof}[Proof of Lemma~\ref{lem:BC<A}]
Under setting (A),  $\matY$ satisfies $(\e, \delta)$-DP Definition~\ref{def:DiffPriv}. That is, for any set $\mathcal{S}$ and any pair of neighbors $\matX$ and $\matX'$,
$$
P_\matX[\matY \in \mathcal{S}] = P_\matX[(\matX + \matC) \in \mathcal{S}] \le  e^\e P_{\matX'}[\matY \in \mathcal{S}] + \delta = e^\e P_{\matX'}[(\matX' + \matC) \in \mathcal{S}] + \delta.
$$
Let $\nu_\sigma(\cdot)$ denote the multivariate Gaussian density for $\matC \sim NI_{n \times p}(0, \sigma^2)$, and $\mathbb{I}(\mathcal{E})$ denote the indicator variable that event $\mathcal{E}$ occurs. Then the above expression becomes
\be\label{eq:A.DP}
\int \mathbb{I}[(\matX + \matc) \in \mathcal{S}] d \nu_{\sigma_0}(\matc)  \le  e^\e \int \mathbb{I}[(\matX' + \matc) \in \mathcal{S}] d \nu_{\sigma_0}(\matc) + \delta,
\ee
for any set $\mathcal{S}$ and any pair of neighbors $\matX$ and $\matX'$.

Under setting (B), 
denote $\matA \mathcal{S} =\{ \matA \matx: \matx \in \mathcal{S}\}$ as the set whose elements are the elements in $\mathcal{S}$ multiplied by the matrix $\matA$. Then
$$
\ba{cl}
P_\matX[\matY \in \mathcal{S}] & = \int \int \mathbb{I}[\matA(\matX + \matc) \in \mathcal{S}] d \nu_{\sigma_0}(\matc) d \mu(\matA) \\
& = \int \int \mathbb{I}[(\matX + \matc) \in \matA^{-1} \mathcal{S}] d \nu_{\sigma_0}(\matc) d \mu(\matA).\\
\ea
$$
When $\sigma = \sigma_0$, this together with \eqref{eq:A.DP} implies that for any set $\mathcal{S}$ and any pair of neighbors $\matX$ and $\matX'$,
$$
\ba{cl}
P_\matX[\matY \in \mathcal{S}] & \le \int \{e^\e \int \mathbb{I}[(\matX' + \matc) \in \matA^{-1} \mathcal{S}] d \nu_{\sigma_0}(\matc) + \delta\} d \mu(\matA) \\
& = e^\e \int \int \mathbb{I}[(\matX' + \matc) \in \matA^{-1} \mathcal{S}] d \nu_{\sigma_0}(\matc) d \mu(\matA) + \int \delta d \mu(\matA) \\
& = e^\e P[(\matX' + \matC) \in \mathcal{S}] + \delta \\
& = e^\e P_{\matX'}[\matY \in \mathcal{S}] + \delta.
\ea
$$
Thus the mechanism in setting (B) also satisfies $(\e, \delta)$-DP Definition~\ref{def:DiffPriv} for $\sigma = \sigma_0$.

\end{proof}

\subsection{Proof of Lemma~\ref{lem:DPdensity}}\label{sec:proof.DPdensity}

\begin{proof}[Proof of Lemma~\ref{lem:DPdensity}]
Let $\mathcal{A}^C$ denotes the compliment of a set $\mathcal{A}$.
When condition~\eqref{eq:DP-bound0} holds, for any set $\mathcal{S}$ and any pair of neighbors $\matX$ and $\matX'$,
$$
\ba{rl}
P_\matX[Y \in \mathcal{S}] = & P_\matX[Y \in \mathcal{S} \cap \mathcal{S}_{\matX,\matX'}^C] + P_\matX[Y \in \mathcal{S} \cap \mathcal{S}_{\matX,\matX'}] \\
\le & \int_{\mathcal{S} \cap \mathcal{S}_{\matX,\matX'}^C} p_{\matX}(\maty) d \maty + P_\matX[Y \in \mathcal{S}_{\matX,\matX'}] \\
\le & e^\e \int_{\mathcal{S} \cap \mathcal{S}_{\matX,\matX'}^C} p_{\matX'}(\maty) d\maty + \delta \\
\le & e^\e \int_{\mathcal{S}} p_{\matX'}(\maty) d\maty + \delta \\
= & e^\e P_{\matX'}[Y \in \mathcal{S}] + \delta.
\ea
$$
Thus the mechanism $Y$ satisfies $(\e, \delta)$-DP Definition~\ref{def:DiffPriv}.

When condition~\eqref{eq:DP-bound0} is violated, there exists a pair of neighboring $\matX$ and $\matX'$ such that $P_\matX[Y \in \mathcal{S}_{\matX,\matX'}] > \delta$. Take $\mathcal{S} = \mathcal{S}_{\matX,\matX'}$, then
$$
P_{\matX'}[Y \in \mathcal{S}] = \int_{\mathcal{S}} p_{\matX'}(\maty) d\maty \le \frac{1}{e^\e} \int_{\mathcal{S}} p_{\matX}(\maty) d\maty = \frac{1}{e^\e} P_{\matX}[Y \in \mathcal{S}].
$$
Hence,
$$
\ba{cl}
e^{\e'} P_{\matX'}[Y \in \mathcal{S}] + \delta' &=e^{\e'} P_{\matX'}[Y \in \mathcal{S}] + (1-\frac{e^{\e'}}{e^\e})\delta \\
 & \le  \frac{e^{\e'}}{e^\e} P_{\matX}[Y \in \mathcal{S}] + (1-\frac{e^{\e'}}{e^\e})\delta \\
&< \frac{e^{\e'}}{e^\e} P_{\matX}[Y \in \mathcal{S}] + (1-\frac{e^{\e'}}{e^\e})P_{\matX}[Y \in \mathcal{S}_{\matX,\matX'}] \\
&= \frac{e^{\e'}}{e^\e} P_{\matX}[Y \in \mathcal{S}] + (1-\frac{e^{\e'}}{e^\e})P_{\matX}[Y \in \mathcal{S}] \\
&= P_{\matX}[Y \in \mathcal{S}].
\ea
$$
Thus the mechanism is not $(\e', \delta')$-DP.
\end{proof}

\subsection{Proof of Corollary~\ref{cor:bound.settingA}}\label{sec:proof.cor.A}

\begin{proof}[Proof of Corollary~\ref{cor:bound.settingA}]
We will use an explicit bound \eqref{eq:GaussQuantBound} on the Gaussian quantile provided in the next Lemma~\ref{lem:GaussianQuantile}.

By Theorem~\ref{thm:bound.settingA}, if \eqref{eq:DP-bound0} holds for every pair of neighboring $\matX$ and $\matX'$, it is necessary that
$$
\sigma \ge \frac{\bar \gamma_{\delta}}{\e} > \frac{\sqrt{ln(\frac{1}{\delta})}}{\e},
$$
where the last inequality comes from \eqref{eq:GaussQuantBound} $\sqrt{\ln(\frac{1}{\delta})} < \bar \gamma_{\delta}$.

Furthermore, when $\delta<0.05$, $2 ln(\frac{1}{\delta})>2 ln(20)>5$, thus $(1+\frac{1}{2 \ln(\frac{1}{\delta})}) < 1.2$. Hence
$$
\sigma > \frac{1.7  \sqrt{ \ln(\frac{1}{\delta})}}{\e} > \frac{1.2  \sqrt{2 \ln(\frac{1}{\delta})}}{\e} \qquad \Rightarrow \qquad \sigma > \frac{ \sqrt{2 \ln(\frac{1}{\delta})}}{\e} (1+\frac{1}{2 \ln(\frac{1}{\delta})}) > \frac{\bar \gamma_{\delta}}{\e}(1+\frac{1}{2 \bar \gamma_{\delta}^2})),
$$
where the last inequality comes from \eqref{eq:GaussQuantBound} $\sqrt{\ln(\frac{1}{\delta})} < \bar \gamma_{\delta} < \sqrt{2 \ln(\frac{1}{\delta})}$. Hence by Theorem~\ref{thm:bound.settingA}, condition~\eqref{eq:DP-bound0} holds for every pair of neighboring $\matX$ and $\matX'$ (thus the release mechanism achieves $(\e, \delta)$-DP).

\end{proof}

\begin{lemma}\label{lem:GaussianQuantile}
When $\delta<0.05$,
\be\label{eq:GaussQuantBound}
\sqrt{\ln(\frac{1}{\delta})} < \bar \gamma_{\delta} < \sqrt{2 \ln(\frac{1}{\delta})}
\ee
\end{lemma}

\begin{proof}[Proof of Lemma~\ref{lem:GaussianQuantile}]
$\bar \gamma_{\delta}$ can be found using the Gaussian tail bound: for  $Z \sim N(0,1)$ and $t>0$,
$$
\frac{e^{-\frac{t^2}{2}}}{\sqrt{2 \pi}} (\frac{1}{t}-\frac{1}{t^3}) < P[Z>t] <\frac{e^{-\frac{t^2}{2}}}{\sqrt{2 \pi}} \frac{1}{t}.
$$
Hence
$$
\frac{e^{-\frac{\bar \gamma_{\delta}^2}{2}}}{\sqrt{2 \pi}} \frac{1}{\bar \gamma_{\delta}}\frac{\bar \gamma_{\delta}^2-1}{\bar \gamma_{\delta}^2}= \frac{e^{-\frac{\bar \gamma_{\delta}^2}{2}}}{\sqrt{2 \pi}} (\frac{1}{\bar \gamma_{\delta}}-\frac{1}{\bar \gamma_{\delta}^3}) < \qquad \delta= P[Z>\bar \gamma_{\delta}] \qquad <\frac{e^{-\frac{\bar \gamma_{\delta}^2}{2}}}{\sqrt{2 \pi}} \frac{1}{\bar \gamma_{\delta}}.
$$
Taking the natural logarithm on all sides, we get
$$
-\frac{\bar \gamma_{\delta}^2}{2} - \ln(\sqrt{2 \pi}) - \ln(\bar \gamma_{\delta}) - \ln(\frac{\bar \gamma_{\delta}^2}{\bar \gamma_{\delta}^2-1}) < \qquad \ln(\delta)\qquad < -\frac{\bar \gamma_{\delta}^2}{2} - \ln(\sqrt{2 \pi}) - \ln(\bar \gamma_{\delta}).
$$
Since $\ln(\frac{1}{\delta}) = -\ln(\delta)$,
\be\label{eq:gamma_deltaBound}
\frac{\bar \gamma_{\delta}^2}{2} + \ln(\sqrt{2 \pi}) + \ln(\bar \gamma_{\delta})   < \qquad \ln(\frac{1}{\delta}) \qquad < \frac{\bar \gamma_{\delta}^2}{2} + \ln(\sqrt{2 \pi}) + \ln(\bar \gamma_{\delta})  + \ln(\frac{\bar \gamma_{\delta}^2}{\bar \gamma_{\delta}^2-1}).
\ee

When $\delta<0.05$, $\bar \gamma_{\delta} > \bar \gamma_{0.05} = 1.645 >1$, hence $\ln(\sqrt{2 \pi}) + \ln(\bar \gamma_{\delta})>0$. The left inequality in \eqref{eq:gamma_deltaBound} implies that
$$
\qquad \frac{\bar \gamma_{\delta}^2}{2} < \frac{\bar \gamma_{\delta}^2}{2} + \ln(\sqrt{2 \pi}) + \ln(\bar \gamma_{\delta})   <  \qquad\ln(\frac{1}{\delta}).
$$
That is, $\bar \gamma_{\delta} < \sqrt{2 \ln(\frac{1}{\delta})}$ which is the second half of \eqref{eq:GaussQuantBound}.

Let $f(x) = - \frac{x^2}{2} + \ln(\sqrt{2 \pi}) + \ln(x)  + \ln(\frac{x^2}{x^2-1})$. Then the derivative
$$
f'(x) = -x + \frac{1}{x} + \frac{-2}{x(x^2-1)}  <0 \qquad \mbox{for all } x>1.
$$
Hence $f(x)$ is an decreasing function when $x>1$. Therefore, when $\delta<0.05$, $f(\bar \gamma_{\delta}) < f(\bar \gamma_{0.05}) = f(1.645) < 0$. Apply this to the right inequality in \eqref{eq:gamma_deltaBound},
$$
\ln(\frac{1}{\delta}) \qquad < \frac{\bar \gamma_{\delta}^2}{2} + \ln(\sqrt{2 \pi}) + \ln(\bar \gamma_{\delta})  + \ln(\frac{\bar \gamma_{\delta}^2}{\bar \gamma_{\delta}^2-1}) \qquad = \bar \gamma_{\delta}^2 + f(\bar \gamma_{\delta}) \qquad < \bar \gamma_{\delta}^2.
$$
This is the first half of \eqref{eq:GaussQuantBound}.
\end{proof}

\subsection{Proof of Corollary~\ref{cor:bound.settingB_Ver3}}\label{sec:proof.cor.B}
\begin{proof}[Proof of Corollary~\ref{cor:bound.settingB_Ver3}]

From \eqref{eq:sigma.bound.new1} $\sigma \ge  \sqrt{\frac{2n-p+\ln(\frac{1}{\delta})}{2(n-p)}} \frac{3\sqrt[4]{p}}{\sqrt{\e}}$, we have
\be\label{eq:sigma2ep.ge}
\ba{cl}
\sigma^2 \e \ge \frac{9\sqrt{p}[(2n-p)+\ln(\frac{1}{\delta})]}{2(n-p)} & = \frac{3\sqrt{p}[2(2n-p)+3\ln(\frac{1}{\delta})]+ 3\sqrt{p}(2n-p)}{2(n-p)} \\
& \ge  \frac{(2\sqrt{p}+1)[2(2n-p) +3\ln(\frac{1}{\delta})]}{2(n-p)} + \frac{2\sqrt{p}(2n-p)}{2(n-p)} \\
& \ge \frac{2\sqrt{p}+1}{2(n-p)} [2(2n-p) +3\ln(\frac{1}{\delta})] + \sqrt{p} \\
& \ge \frac{(2 \sqrt{p} +1)  }{2 (n-p) } \gamma_{\delta; 2(n-p), \sqrt{p}}  + \sqrt{p},
\ea
\ee
where the last inequality comes from bound \eqref{eq:gamma_delta} on the Chi-square quantile provided in the Lemma~\ref{lem:gamma_delta_Bound}.

The first term in the last expression is always positive, thus \eqref{eq:sigma2ep.ge} implies that $\sigma^2 \e \ge \sqrt{p}$ which then implies that $\frac{\sqrt{p}}{\sigma^2} \le \e \le 1$. Hence $\frac{p}{\sigma^2} \le \sqrt{p}$. Therefore, $\gamma_{\delta; 2(n-p), \frac{p}{\sigma^2}} \le \gamma_{\delta; 2(n-p), \sqrt{p}}$ because the upper quantile of the Chi-square distribution increases when the non-central parameter increases. Plug this back into \eqref{eq:sigma2ep.ge}, we have
$$
\sigma^2 \e  \ge \frac{(2 \sqrt{p} +1)  }{2 (n-p) } \gamma_{\delta; 2(n-p), \sqrt{p}}  + \sqrt{p} \ge \frac{(2 \sqrt{p} +1)  }{2 (n-p) } \gamma_{\delta; 2(n-p), \frac{p}{\sigma^2}}  + \sqrt{p}.
$$

That is, \eqref{eq:sigma.bound.new} holds. Hence Theorem~\ref{thm:bound.settingB_Ver3} ensures that the mechanism in setting (B) is $(\e, \delta)$-DP.

This finishes the proof.
\end{proof}

\subsection{Chi-square distribution Tail bound.}
 We first cite the following lemma which is the Lemma8.1 in \cite{birge2001alternative}.
\begin{lemma}\label{lem:Chi-tail}
If $\chi^2$ follows a Chi-square distribution with $n$ degrees of freedom and noncentral parameter $\nu$ then for any $x>0$,
\be\label{eq:Chi-tail}
\ba{cl}
&P[\chi^2 \ge (n+\nu^2) + 2\sqrt{(n+2 \nu^2)x} + 2x ] \le e^{-x}, \\
&P[\chi^2 \le (n+\nu^2) - 2\sqrt{(n+2 \nu^2)x} ] \le e^{-x}.
\ea
\ee
\end{lemma}

Using these bounds, we can find a simple bound for $\gamma_{\delta; 2(n-p), \sqrt{p}}$ as in the following Lemma.
\begin{lemma}\label{lem:gamma_delta_Bound}
\be\label{eq:gamma_delta}
\gamma_{\delta; 2(n-p), \sqrt{p}} \le 2(2n-p)+3ln(\frac{1}{\delta}).
\ee
\end{lemma}

\begin{proof}[Proof of Lemma~\ref{lem:gamma_delta_Bound}]
Plug $x=ln(\frac{1}{\delta})$  into \eqref{eq:Chi-tail} for the Chi-square distribution with $2(n-p)$ degrees of freedom and noncentral parameter $\nu=\sqrt{p}$,
$$
P\left[\chi^2 \ge [2(n-p) + p ] + 2\sqrt{[2(n-p) + p ]ln(\frac{1}{\delta})} + 2ln(\frac{1}{\delta}) \right]  \le e^{-ln(\frac{1}{\delta})} = e^{ln(\delta)} = \delta.
$$
Since by definition,
$$
P[\chi^2 \ge  \gamma_{\delta; 2(n-p), \sqrt{p}}] = \delta,
$$
we have
$$
\ba{cl}
\gamma_{\delta; 2(n-p), \sqrt{p}} & \le [2(n-p) + p ] + 2\sqrt{[2(n-p) + p ]ln(\frac{1}{\delta})} + 2ln(\frac{1}{\delta}) \\
& \le [2(n-p) + p ] + \{[2(n-p) + p ] + ln(\frac{1}{\delta})\} + 2ln(\frac{1}{\delta}) \\
& = 2(2n-p) + 3ln(\frac{1}{\delta}).
\ea
$$
\end{proof}

\subsection{Proof of Lemma \ref{lem:int_unif}}\label{sec:proof.unif}
\begin{proof}[Proof of Lemma~\ref{lem:int_unif}]
 Let $\mate_1 = \frac{\bm{proj}_\mathfrak{S} (\matv)}{\| \bm{proj}_\mathfrak{S} (\matv) \|}$ be the unit vector along the direction of the projection for $\matv$ onto the subspace $\mathfrak{S}$. Then $\matb \matv^T = \matb \mate_1^T \| \bm{proj}_\mathfrak{S} (\matv) \|$.

  For the subspace $\mathfrak{S}$, we can find unit vectors $\mate_1, ..., \mate_q$ which are orthogonal to each other. They then form the base vectors for a coordinate systems of $\mathfrak{S}$. Hence $\matb = \sum_{i=1}^{q} \tilde{b}_i \mate_i$ with $\tilde{b}_i = \matb \mate_i^T $ being the $i$-th coordinate of $\matb$ under this coordinate system. Hence
  $$
  \matb \matv^T = \tilde{b}_1 \| \bm{proj}_\mathfrak{S} (\matv) \|.
  $$

  Notice $\tilde{b}_1$ is the first coordinate of a random vector uniformly distributed over the $(q-1)$-dimensional unit sphere, and its probability density is known to be $\frac{1}{\bar c_q} (1-u^2)^\frac{q-3}{2} $. Thus
  $$
  \ba{cl}
  \int g( \matb \matv^T) d \mu_q(\matb) & = \int g( \| \bm{proj}_\mathfrak{S} (\matv) \| \tilde{b}_1 ) d \mu_q(\matb) \\
  & =  \int_{u=-1}^{1} g( \|\bm{proj}_\mathfrak{S} (\matv) \| u) \frac{1}{\bar c_q} (1-u^2)^\frac{q-3}{2} du.
  \ea
  $$

\end{proof}

\subsection{Proof of Lemma \ref{lem:GRatio}}\label{sec:proof.GRatio}

\begin{proof}[Proof of Lemma~\ref{lem:GRatio}]
$$
\ba{cl}
G_q'(t) = \int_{u=-1}^{1} e^{ t u} u (1-u^2)^\frac{q-2}{2} du  = \int_{u=0}^{1} [e^{ t u} - e^{ - tu}] u (1-u^2)^\frac{q-2}{2} du.
\ea
$$
Since $e^{ t u} - e^{ - tu} >0$ for all positive $tu$ values, the integrand is always positive in the last integral. Hence
$G_q'(t) >0,$
when $t>0$. Thus $G_q(t)$ is an increasing function when $t>0$.

Furthermore, using integral by parts,
$$
\ba{cl}
G_q'(t) & = \int_{u=-1}^{1} e^{ t u} u (1-u^2)^\frac{q-2}{2}  \\
& = \int_{u=-1}^{1} e^{ t u}  \frac{1}{2} (1-u^2)^\frac{q-2}{2} d(u^2) \\
& = \int_{u=-1}^{1} e^{ t u}  \frac{-1}{q} d[(1-u^2)^\frac{q}{2}] \\
& =  \frac{-e^{ t u}}{q} (1-u^2)^\frac{q}{2} |_{u=-1}^{1} -  \int_{u=-1}^{1} \frac{-1}{q} (1-u^2)^\frac{q}{2} d(e^{ t u}) \\
& =  0 -  \int_{u=-1}^{1} \frac{-1}{q} (1-u^2)^\frac{q}{2} e^{ t u} t du \\
& =  \frac{t}{q}  \int_{u=-1}^{1} e^{ t u} (1-u^2)^\frac{q}{2}  du \\
& \le  \frac{t}{q}  \int_{u=-1}^{1} e^{ t u} (1-u^2)^\frac{q-2}{2}  du \; = \frac{t}{q} G_q(t).
\ea
$$

For any $t_2 > t_1 >0$, since $G_q(t) > 0$, we have
$$
\ba{cl}
\ln G_q(t_2) - \ln G_q(t_1) = \int_{t=t_1}^{t_2} \frac{G_q'(t)}{G_q(t)} dt  \le \int_{t=t_1}^{t_2} \frac{t}{q}  dt = \frac{t_2^2-t_1^2}{2q}.
\ea
$$

Hence
$$
\frac{G_q(t_2)}{G_q(t_1)}  \le e^{\frac{|t_2^2-t_1^2|}{2q}},
$$
which is also true when $t_2 < t_1$ because $G_q(t)$ is increasing.

\end{proof}

\subsection{Probability Density under setting (C)}\label{sec:DenC}
\begin{lemma}\label{lem:DenC}
Under setting (C): $\matY =  \matA \matX + \matC$, we have
\be\label{eq:DenC}
p_{\matX}(\maty) = (\frac{1}{\sqrt{2\pi}\sigma})^{np} e^{-\frac{\|\maty\|^2}{2 \sigma^2}} e^{-\frac{\|\matX  \|^2}{2 \sigma^2}} \int_{\matA \in \mathcal{O}_{n \times n}}  e^{\frac{tr(\matA^T\maty\matX^T)}{ \sigma^2}} d \mu(\matA).
\ee
\end{lemma}

\begin{proof}[Proof of Lemma~\ref{lem:DenC}]
From \eqref{eq:modelC}, $\matC = \matY- \matA \matX$. Hence
$$
p_{\matX}(\maty)  = \int_{\matA \in \mathcal{O}_{n \times n}} (\frac{1}{\sqrt{2\pi}\sigma})^{np} e^{-\frac{\| \maty - \matA \matX\|^2}{2 \sigma^2}} d \mu(\matA).
$$
Since
$$
\ba{cl}
\| \maty - \matA \matX\|^2 = \|\maty\|^2 + \| \matA \matX\|^2 - 2 tr[\maty(\matA\matX)^T] & = \|\maty\|^2 + \| \matX\|^2 - 2 tr[\maty\matX^T\matA^T] \\
& = \|\maty\|^2 + \| \matX\|^2 - 2 tr(\matA^T\maty\matX^T),
\ea
$$
the probability density above becomes
$$
p_{\matX}(\maty) = \int_{\matA \in \mathcal{O}_{n \times n}}  (\frac{1}{\sqrt{2\pi}\sigma})^{np} e^{-\frac{\|\maty\|^2 + \|\matX\|^2 - 2tr(\matA^T\maty\matX^T)}{2 \sigma^2}} d \mu(\matA).
$$
Thus the density is the same as in \eqref{eq:DenC}. Notice that this is also the same as the probability density \eqref{eq:DenB} in setting (B).
\end{proof}

\end{appendix}

\section*{Acknowledgment}
This work is supported in part by NIH grants R01LM014027 and U24 AA029959-01.

\bibliography{references}
\bibliographystyle{abbrvnat}

\end{document}